\documentclass
[superscriptaddress,secnumarabic,amssymb,amsmath,nobibnotes,aps,prd,showkeys,nopacs,nofootinbib,onecolumn,12pt]{revtex4}%
\usepackage{graphicx}
\usepackage{epsf}
\usepackage{bm}
\usepackage{amsmath}
\usepackage{amsfonts}
\usepackage{amssymb}
\usepackage{amsthm}
\usepackage{amsmath,amssymb,graphicx}
\usepackage{mathrsfs}
\usepackage{mwe}
\usepackage{enumitem}
\usepackage{leftindex}

\newtheorem{Proposition}{Proposition}

\newtheorem{Remark}{Remark}

\usepackage{tikz,xcolor,hyperref}

\definecolor{lime}{HTML}{A6CE39}
\DeclareRobustCommand{\orcidicon}{%
	\begin{tikzpicture}
	\draw[lime, fill=lime] (0,0) 
	circle [radius=0.16] 
	node[white] {{\fontfamily{qag}\selectfont \tiny ID}};
	\draw[white, fill=white] (-0.0625,0.095) 
	circle [radius=0.007];
	\end{tikzpicture}
	\hspace{-2mm}
}

\foreach \x in {A, ..., Z}{%
	\expandafter\xdef\csname orcid\x\endcsname{\noexpand\href{https://orcid.org/\csname orcidauthor\x\endcsname}{\noexpand\orcidicon}}
}

\begin{document}

% Title
\title{Fractional Time-Delayed differential equations: Applications in Cosmological Studies}
\author{Bayron Micolta-Riascos\orcidA{}}
\email{bayron.micolta@alumnos.ucn.cl}
\affiliation{Departamento de Física, Universidad Católica del Norte, Avda. Angamos 0610, Casilla 1280, Antofagasta, Chile}
\author{Byron Droguett\orcidB{}}
\email{byron.droguett@uantof.cl}
\affiliation{Department of Physics, Universidad de Antofagasta, 1240000 Antofagasta, Chile}
\author{Gisel Mattar Marriaga }
\email{gs618552@dal.ca}
\affiliation{Department of Mathematics and Statistics, Dalhousie University,
Halifax, Nova Scotia, Canada B3H 3J5}
\author{Genly Leon\orcidC{}}
\email{genly.leon@ucn.cl}
\affiliation{Departamento  de  Matem\'aticas,  Universidad  Cat\'olica  del  Norte, Avda. Angamos  0610,  Casilla  1280  Antofagasta,  Chile}
\affiliation{Department of Mathematics, Faculty of Applied Sciences, Durban University of Technology, Durban 4000, South Africa}
\author{Andronikos Paliathanasis\orcidD{}}
\email{anpaliat@phys.uoa.gr}
\affiliation{Departamento  de  Matem\'aticas,  Universidad  Cat\'olica  del  Norte, Avda. Angamos  0610,  Casilla  1280  Antofagasta,  Chile}
\affiliation{Department of Mathematics, Faculty of Applied Sciences, Durban University of Technology, Durban 4000, South Africa}
\affiliation{School for Data Science and Computational Thinking and Department of
Mathematical Sciences, Stellenbosch University, Stellenbosch, 7602, South
Africa}
\author{Luis del Campo\orcidE{}}
\email{lcampo@ucn.cl }
\affiliation{Departamento  de  Matem\'aticas,  Universidad  Cat\'olica  del  Norte, Avda. Angamos  0610,  Casilla  1280  Antofagasta,  Chile}
\author{Yoelsy Leyva\orcidF{}}
\email{yoelsy.leyva@academicos.uta.cl }
\affiliation{Departamento  de  Física, Facultad de Ciencias, Universidad de Tarapac\'a, Casilla 7-D, Arica, Chile}

% Abstract
\begin{abstract}
Fractional differential equations model processes with memory effects, providing a realistic perspective on complex systems. We examine time-delayed differential equations, discussing first-order and fractional Caputo time-delayed differential equations. We derive their characteristic equations and solve them using the Laplace transform. We derive a modified evolution equation for the Hubble parameter incorporating a viscosity term modeled as a function of the delayed Hubble parameter within Eckart’s theory. We extend this equation using the last-step method of fractional calculus, resulting in Caputo's time-delayed fractional differential equation. This equation accounts for the finite response times of cosmic fluids, resulting in a comprehensive model of the Universe's behavior. We then solve this equation analytically. Due to the complexity of the analytical solution, we also provide a numerical representation. Our solution reaches the de Sitter equilibrium point. Additionally, we present some generalizations.
\end{abstract}

\keywords{Fractional calculus; dynamical systems; Caputo time-delayed differential equations; modified gravity.} 

\maketitle

%%%%%%%%%%%%%%%%%%%%%%%%%%%%%%%%%%%%%%%%%%
\section{Introduction}
\noindent

Differential equations play a crucial role in modeling cosmic evolution, particularly expansion and large-scale structure \cite{kolb1990early, ellis1999dynamics}. They are essential for understanding galaxy formation, black holes, and other astrophysical phenomena \cite{wainwrightellis1997, wald1984general, frankel2011the}. Additionally, dynamical systems facilitate the study of chaos and stability in an expanding universe, revealing how small perturbations affect large-scale dynamics \cite{wainwrightellis1997}.  

The standard cosmological paradigm assumes large-scale homogeneity and isotropy, simplifying perturbation studies within Friedmann-Lemaître-Robertson-Walker (FLRW) metrics. While useful, this assumption does not necessarily account for early-universe anisotropies or local inhomogeneities. Researchers have explored whether inflation can generate homogeneity and isotropy from more general metrics, including inhomogeneous and anisotropic ones \cite{Goldwirth:1989pr}. Due to the complexity of such approaches, studies increasingly focus on homogeneous but anisotropic cosmologies, which offer valuable analytical insights.  

Beyond the FLRW model, alternative models incorporate anisotropies and inhomogeneities. Bianchi universes permit anisotropic expansion, where different spatial directions evolve at varying rates. These models enhance our understanding of early-universe dynamics, gravitational waves, and Cosmic Microwave Background (CMB) anomalies.  

The family of spatially homogeneous Bianchi cosmologies includes key gravitational models, such as the Mixmaster Universe and isotropic FLRW spacetimes \cite{Ryan:1975jw, Misner:1967uu, Misner:1969hg, PhysRevD.55.7489}. The nine anisotropic Bianchi models, classified by three-dimensional real Lie algebra, define homogeneous hypersurfaces through isometry group actions. Their physical variables depend only on time, reducing Einstein's field equations to a system of ordinary differential equations \cite{Ellis:1968vb, Goliath:1998na}.  

FLRW spacetimes emerge as limiting cases of Bianchi models when anisotropy vanishes. The flat, open, and closed FLRW geometries correspond to Bianchi I, III, and IX spacetimes, respectively \cite{wainwrightellis1997}. While general Bianchi spacetimes feature three scale factors \cite{Ryan:1975jw}, locally rotationally symmetric (LRS) spacetimes introduce an additional isometry, reducing the number of independent scale factors to two. The LRS Bianchi IX spacetime is closely related to the Kantowski-Sachs geometry \cite{Kantowski:1966te}.  

On the other hand, Lemaître-Tolman-Bondi (LTB) universes relax the assumption of homogeneity by introducing variations in density and curvature across cosmic regions. This framework explains cosmic voids and non-uniform matter distributions, refining our approach to cosmological evolution.
LTB models \cite{Bondi:1947fta, Lemaitre:1933gd, Krasiński_1997} represent spherically symmetric dust solutions to Einstein's equations, generalizing FLRW cosmology. LTB metrics with a dust source and a comoving, geodesic four-velocity form a well-known class of exact solutions \cite{Krasiński_1997, Stephani_Kramer_MacCallum_Hoenselaers_Herlt_2003}. Structure formation in LTB models has been explored in \cite{Krasiński_1997, Krasinski:1995zz, Krasinski:2003yp}. Recent studies analyze inhomogeneous dust solutions numerically and as three-dimensional dynamical systems, using an averaged density parameter $\langle \Omega \rangle$, a shear parameter, and a density contrast function to quantify the effects of inhomogeneity \cite{Sussman:2007ku, Sussman:2008wx}. The evolution equations for these averaged variables formally mirror those of FLRW cosmology.  

While many alternative models provide insights into cosmic evolution, FLRW cosmology remains the dominant framework due to its simplicity and effectiveness in describing large-scale phenomena.  

The Standard Model of Cosmology, known as $\Lambda$CDM, consists of $\Lambda$ (Dark Energy), which drives accelerated expansion, and CDM (Cold Dark Matter), an unseen component interacting only through gravity. This model successfully explains the observed late-time acceleration of the Universe, initially indicated by Type Ia supernova (SNIa) observations~\cite{Riess:1998} and later confirmed by CMB measurements~\cite{Planck:2018}. It also accounts for the formation of large-scale cosmic structures.  

Despite its success, $\Lambda$CDM faces several theoretical challenges, such as the cosmological constant problem~\cite{Zeldovich, Weinberg}, the nature of dark matter and dark energy, the origin of accelerated expansion~\cite{Carroll:2000}, and the Hubble tension~\cite{Riess2021HubbleTension, DiValentino:2025sru}.  

While the flatness and horizon problems can be theoretically resolved through inflation, the fundamental cause of inflation remains uncertain~\cite{Guth1981Inflation, Linde1982newInflation}. Various alternative theories attempt to address these issues, including noncommutative theories, quantum cosmology, quantum-deformed phase space models, and noncommutative minisuperspace approaches~\cite{Rasouli:2013sda, Jalalzadeh:2014jea, Rasouli:2014dba, Rasouli:2016syh, Jalalzadeh:2017jdo}, as well as modified Brans-Dicke theory~\cite{Rasouli:2014dxa, Rasouli:2016xdo}.

Traditional mathematical models are often inadequate for describing power-law phenomena, which exhibit frequency-dependent, non-local, and history-dependent characteristics. Fractional calculus provides a mathematical framework to address these challenges by extending differentiation and integration to non-integer orders. Unlike classical derivatives, fractional derivatives consider the complete historical behavior of a system, making them particularly suited for applications where past states influence present dynamics.
This approach has proven effective in modeling systems with frequency-dependent properties, such as viscoelastic materials and electrical circuits. Its broad applicability spans disciplines like quantum physics, engineering, biology, and finance. Researchers have used fractional calculus to investigate complex topics such as quantum fields~\cite{Lim:2006hp, LimEab+2019+237+256}, quantum gravity~\cite{El-nabulsi:2013hsa, El-nabulsi:2013mwa}, black holes~\cite{Vacaru:2010wn, Jalalzadeh:2021gtq}, and cosmology~\cite{Moniz:2020emn, Rasouli:2021lgy, VargasMoniz:2020hve}. Using dynamical system methods combined with observational data testing provides a robust framework for analyzing the physical behavior of cosmological models. This approach has led to the development of cosmological models that exhibit late-time acceleration without requiring dark energy. Key studies include joint analysis using cosmic chronometers (CCs) and Sne Ia data to determine best-fit values for fractional order derivatives~\cite{Garcia-Aspeitia:2022uxz}, improved observational tests in subsequent studies~\cite{Gonzalez:2023who, LeonTorres:2023ehd}, and deduced equations of state for a matter component based on compatibility conditions~\cite{Micolta-Riascos:2023mqo}. 

Researchers have developed fractional versions of traditional Newtonian mechanics and Friedmann-Robertson-Walker cosmology by incorporating fractional derivatives into the equations. Examples include non-local-in-time fractional higher-order Newton's second law of motion~\cite{ELnABULSI201765} and fractional dynamics exhibiting disordered motions. Two primary methods for developing fractional derivative methods have emerged: the last-step modification method, which substitutes original cosmological field equations with fractional field equations tailored for a specific model, as seen in~\cite{Barrientos:2020kfp}, and a more fundamental approach where fractional derivative geometry is established initially, followed by the application of the Fractional Action-like Variational Approach (FALVA)~\cite{El-nabulsi-Torres-2008, Baleanu-Trujillo-2010, Odzijewicz-Malinowska-Torres-2013c}. Recently, fractional cosmology has emerged as a novel explanation for the Universe's accelerated expansion~\cite{LeonTorres:2023ehd, Micolta-Riascos:2023mqo, Rasouli2024, Barrientos2021}, utilizing both the first-step and last-step methods to achieve results consistent with cosmological observations.

On the other hand, viscous cosmology models cosmic fluids by accounting for dissipative effects, incorporating dissipation terms through Eckart's or Israel-Stewart's theories~\cite{DASILVA201937, Tamayo_2022, Cruz2020, Lepe2017}. These terms can be introduced as effective pressure in the energy-momentum tensor, modifying the Friedmann and continuity equations. Viscous cosmology has applications in the early Universe, where it can drive inflation without requiring a scalar field, as well as in late-time cosmology to model the Universe's accelerated expansion~\cite{Maartens1996Dissipation, Brevik2017}. In cosmology, time delay has also been integrated into the field equations to model the finite response time of the gravitational system to perturbations. Cosmic fluids do not adapt instantaneously; they respond to past cumulative processes, offering a more realistic depiction of these systems. These delay effects stem from non-local interactions in fundamental theories of quantum gravity, which incorporate memory effects into the universe's evolution. A delay term in the Friedmann equation has been proposed to model the inflationary epoch without using a scalar field, sidestepping the violation of the strong energy condition and providing a natural conclusion to the inflationary period \cite{Choudhury2012}. Additionally, applying that delayed Friedmann equation for late-time cosmology was examined, demonstrating that the delay is statistically consistent with the Hubble expansion rate and growth data \cite{PalpalLatoc2021}.

Numerous works on time-delayed differential equations (TDDEs) have been done in the literature. Applying summable dichotomies to functional difference equations focuses on bounded and periodic solutions, offering insights into Volterra systems relevant to biological modeling \cite{del2018bounded}. Summable dichotomies ensure that solutions to these equations remain bounded and periodic \cite{del2012bounded}. Studies on delayed difference equations focus on bounded and periodic solutions, which are significant for systems with delays in engineering and biological models. Nearly periodic solutions are crucial for understanding systems that exhibit regular but not necessarily periodic behavior \cite{del2011almost}. Research on weighted exponential trichotomy and the asymptotic behavior of nonlinear systems helps analyze the long-term behavior of solutions \cite{cuevas2010weighted}. The asymptotic expansion for difference equations with infinite delay provides a framework for approximating solutions to complex systems \cite{cuevas2009asymptotic}. The exploration of weighted exponential trichotomy continues in the comprehensive analysis of linear difference equations \cite{vidal2008weighted}.

TDDEs are crucial for modeling systems where the current rate of change depends on past states, representing processes like incubation periods in infectious diseases and population responses to environmental changes. One study highlights the role of TDDEs in modeling biological processes, including population dynamics and disease spread \cite{https://doi.org/10.1155/2018/4584389}. Another study explores chaotic behavior in diabetes mellitus through numerical modeling of the metabolic system \cite{https://doi.org/10.1155/2018/6815190}. Research on oscillation criteria for delay and advanced differential equations expands the theoretical understanding of these systems \cite{https://doi.org/10.1155/2018/8237634}. Exploring the bifurcations and dynamics of the Rb-E2F pathway, incorporating miR449, sheds light on cell cycle regulation \cite{https://doi.org/10.1155/2017/1409865}.

A method for maximum likelihood inference in univariate TDDE models with multiple delays offers a robust approach to parameter estimation \cite{https://doi.org/10.1155/2017/6148934}. Research on time delay in perceptual decision-making provides insights into brain decision-making processes \cite{https://doi.org/10.1155/2017/4391587}. Further studies on coupled p-Laplacian fractional differential equations with nonlinear boundary conditions contribute to the understanding of fractional calculus \cite{https://doi.org/10.1155/2017/8197610}. The fractional Fredholm integrodifferential equation is solved analytically using the fractional residual power series method \cite{https://doi.org/10.1155/2017/4573589}. A hybrid adaptive pinning control method for synchronizing delayed neural networks with mixed uncertain couplings enhances control strategies \cite{https://doi.org/10.1155/2017/4654020}. A numerical study on a time delay multistrain tuberculosis model of fractional order offers insights into disease dynamics and control \cite{https://doi.org/10.1155/2017/1047384}. Exploring extinction and persistence in a novel delay impulsive stochastic infected predator-prey system with jumps provides a comprehensive analysis of stochastic dynamics in ecological systems \cite{https://doi.org/10.1155/2017/1950970}.

Fractional Time-Delayed Differential Equations (FTDDEs) combine fractional calculus and time delays to model systems, accurately capturing historical effects and delayed reactions. They have diverse applications in scientific and engineering fields. For instance, FTDDEs enhance control systems by designing and analyzing controllers that take response delays into account, leading to more stable and efficient strategies. They also describe the behavior of viscoelastic materials, which exhibit viscous and elastic characteristics with memory effects. Furthermore, FTDDEs model populations that respond with delays to environmental changes, which is crucial for understanding population dynamics and predicting trends. Advanced mathematical tools solve FTDDEs. Laplace transforms convert differential equations into algebraic ones, making solutions more straightforward. Mittag-Leffler Functions generalize the exponential function, providing solutions to fractional differential equations. First-order FTDDEs involve first-order derivatives for modeling straightforward dynamics, while Fractional Caputo Derivative FTDDEs use the Caputo derivative to account for memory effects. Higher-order FTDDEs involve higher-order derivatives to model complex systems with multiple interacting components.

The existence and stability of solutions for time-delayed nonlinear fractional differential equations are essential for ensuring that a differential equation's solution behaves predictably over time, particularly in engineering and natural sciences \cite{geremew2024existence}. A class of Langevin time-delay differential equations with general fractional orders can model complex dynamical systems in engineering, where memory effects and time delays significantly influence the system's behavior \cite{huseynov2021class}. Numerical methods for solving fractional delay differential equations are crucial, especially when analytical solutions are difficult to obtain. These include a finite difference approach \cite{moghaddam2013numerical, ncaputo1}, and a computational algorithm \cite{amin2021computational}. 
Optimal control of nonlinear time-delay fractional differential equations using Dickson polynomials aims to find a control policy that optimizes a specific performance criterion, which is crucial in economics, engineering, and management \cite{chen2021optimal}. Stability and stabilization of fractional order time-delay systems ensure that the system does not exhibit unbounded behavior over time, which is crucial for the safety and reliability of engineering systems \cite{lazarevic2011stability, pakzad2013stability}.
Numerical solutions for multi-order fractional differential equations with multiple delays using spectral collocation methods are known for their high accuracy and efficiency in solving differential equations, making them suitable for complex systems in science and engineering \cite{dabiri2018numerical}.

The global Mittag-Leffler synchronization of discrete-time fractional-order neural networks with time delays ensures that different parts of the network function harmoniously, which is critical for the network's overall performance \cite{zhang2022global}. The stability of oscillators with time-delayed and fractional derivatives is crucial for understanding the behavior of oscillatory systems in electronics, mechanics, and biology \cite{liao2016stability, leung2013fractional}. Stability and control of fractional order time-delay systems are also covered extensively \cite{lazarevic2011stability,luo2011complex,pakzad2013stability} and \cite{liao2016stability,hu2016stability,butcher2017stability}.

In this work, we derive an equation based on the Friedmann and continuity equations, incorporating a viscosity term modeled as a function of the delayed Hubble parameter within Eckart's theory. We then apply the last-step method of fractional calculus to extend this equation, resulting in a fractional delayed differential equation. This framework builds upon the analysis conducted by Paliathanasis in \cite{Paliathanasis:2022wwl}. We solve this equation analytically for the Hubble parameter. Due to the complexity of the analytical solution, we also provide a numerical representation. Additionally, we present the analytical solution of the fractional delayed differential equation that includes \( m \) delayed terms, with the delays multiples of a fundamental delay \( T \). Our solution reaches the de Sitter equilibrium point, generalizing the results in the nonfractional case analyzed in \cite{Paliathanasis:2022wwl}.

The paper is organized as follows. Section \ref{sect.1} presents foundational preliminaries, including an investigation of First-Order Time-Delayed Differential equations and Fractional Caputo Derivative Time-Delayed Differential equations of orders less than one or higher. In Section \ref{sect.3}, we explore a cosmological application—time-delayed bulk viscosity—modeled as a first-order retarded differential equation. Section \ref{sect.4} extends this formulation by promoting it to a fractional version, which is solved analytically. We introduce master and fractional differential equations and outline the problem set—Section \ref{sect.2}. The resulting model represents a time-delayed bulk viscosity within the framework of Fractional Cosmology, with further generalizations incorporating multiple delay scenarios. Finally, Section \ref{sect.5} provides concluding Remarks.
Several appendices based on \cite{mathai2008special,kiryakova2010special,yang2021theory,agarwal2020special} are included to ensure the study's self-contained nature: Appendix \ref{app000} covers the Lambert ($W$) Function,   Appendix \ref{app00} addresses Mittag-Leffler functions,   Appendix \ref{app4} outlines the Laplace transform of the time-delayed function, and  Appendix \ref{app2} discusses the Laplace transform of the Caputo derivative. Numerical Considerations and Forward Difference Formulation are mentioned in Appendix
\ref{appFDF}. In Appendix \ref{appF}, we present optimized algorithms to reproduce our results.

Exploring Caputo fractional differential equations with time delay is essential for viscous cosmology. That helps us understand the Universe better by using more accurate models of cosmic evolution. These equations describe processes with memory effects, showing how complex systems behave. Combined with time delay, they account for the delayed response of cosmic fluids to changes, giving a complete model of the Universe's behavior. These tools can help us understand how viscosity affects cosmic evolution in viscous cosmology. By including memory effects and time-delayed responses, researchers can develop models that show the fundamental physics of the Universe, possibly uncovering new insights into its origins, structure, and future.

\section{Preliminaries}\label{sect.1}

\subsection{First-Order Time-Delayed Differential Equation}
Consider the first-order time-delayed homogeneous differential equation given by:
\begin{equation}
y'(t) + a y(t-T) = 0, \quad y(0)=0, 
\end{equation}
where \(y(t)\) is the dependent variable, \(T>0\) is the time delay, and \(a\) is a constant. 

The characteristic equation for this FTDDE can be written as:
\begin{equation}
s + a e^{-sT} = 0. \label{charat1}
\end{equation}
Equation \eqref{charat1} is a powerful tool for analyzing systems with delays. Stability, oscillations, and parameter-driven dynamics can be explored through the roots and their interactions with the delay term \(e^{-sT}\). The solutions of the characteristic equation are represented by 
\begin{equation}
    s={W(-a T )}/{T }, 
\end{equation}
where $W(z)$ is the Lambert function (see Appendix \ref{app000}). The Lambert $W$ function is used in various fields, such as solving transcendental equations involving exponentials and logarithms, analyzing the behavior of specific dynamical systems, calculating the number of spanning trees in a complete graph, and modeling growth processes and delay differential equations.

The characteristic equation \eqref{charat1} helps analyze the dynamics and stability of delay systems.

\begin{enumerate}
    \item \textbf{Roots of the equation}: Solving for \( s \) yields the roots. Complex roots typically signify oscillatory dynamics, while the real part of each root, $\Re(s)$, plays a critical role in stability analysis. A system is stable if $\Re(s) < 0$ and unstable if $\Re(s) > 0$.
    
    \item \textbf{Delay Effects}: The delay term \( e^{-sT} \) significantly influences root locations and can lead to changes in system stability and behavior, such as bifurcations.

    \item \textbf{Parameter Influence}: The parameter \( a \) affects the feedback within the system, influencing the roots and dynamics. Larger values of \( a \) may amplify feedback effects.
\end{enumerate}

Consider now the inhomogeneous equation 
\begin{equation}
y'(t) + a y(t-T) = b, \quad y(0)\neq 0, \quad y(t)=0, \; t<0 \label{equation(4)}
\end{equation}
where $b$ is a constant.

To solve \eqref{equation(4)} we take the Laplace transform of both sides:
\begin{equation}
s Y(s) - y(0) + a Y(s) e^{-sT} = \frac{b}{s},
\end{equation}
where \(Y(s)\) is the Laplace transform of \(y(t)\). Solving for \(Y(s)\), we get:
\begin{align}
    Y(s) & = \frac{y(0)}{s \left(1 + \frac{a  e^{-sT}}{s}\right)} + \frac{b}{s^2 \left(1 + \frac{a  e^{-sT}}{s}\right)}.\label{eq6}
    \end{align}

\begin{Remark}
\label{rem1}
    Let $c>0$ be an arbitrary constant. Then, the condition $0<\left|\frac{a  e^{-sT}}{c}\right|<1$ is satisfied, if any $c\geq|a|>0$. Indeed, in this case we have $|a/c|\leq 1$ and  
    \[\left|\frac{a  e^{-sT}}{c}\right|=\left|\frac{a}{c}\right|\left| e^{-sT}\right|\leq \left| e^{-sT}\right|<1.\]
\end{Remark}

Using Remark \ref{rem1}, from equation \eqref{eq6} we have 
\begin{equation}
\begin{split}
    Y(s) &  =y(0) \sum_{j=0}^{\infty} (-1)^j \frac{a^j  e^{-s jT}}{s^{j+1}} + b \sum_{j=0}^{\infty} (-1)^j \frac{a^j  e^{-s jT}}{s^{j+2}}, \quad  
   0< \left|\frac{a  e^{-sT}}{s}\right|<1 
\end{split}
\end{equation}
where $s\in (|a|,\infty)$.
The solution $y(t)$ is recovered by applying the inverse Laplace transform
\begin{equation}
\begin{split}
       y(t) & =y(0) \sum_{j=0}^{\infty} (-1)^j  a^j\mathcal{L}^{-1}\left[\frac{e^{-s j T}}{s^{j+1}}\right] + b \sum_{j=0}^{\infty} (-1)^j  a^j\mathcal{L}^{-1}\left[\frac{e^{-s j T}}{s^{j+2}}\right]\\
      & =y(0) \sum_{j=0}^{\infty} (-1)^j  a^j \frac{(t-j T )^j \theta (t-j T )}{\Gamma (j+1)} + b \sum_{j=0}^{\infty} (-1)^j  a^j\frac{(t-j T )^{j+1} \theta (t-j T )}{\Gamma (j+2)}, \label{equation(8)}
\end{split}
\end{equation}
where $\theta$ is the Heaviside Theta 
\[\theta(t)= \begin{cases}
    1, \quad t>0\\
    0, \quad t<0
\end{cases}\]
and 
\begin{equation}
    \Gamma (z)=\int_{0}^{\infty}t^{z-1}e^{-t}\,dt, \quad \Re(z)>0    \label{1.1.1.2}
\end{equation}
is the Gamma function. 

\begin{Remark}\label{Rem_2}
    For each $t>0$, the series in the expression \eqref{equation(8)} is a finite sum. To see this, note that $\theta(t-kT)=0$ for all $k>t/T$. Then, 
    \[
    \begin{split}
    y(t) & =\sum_{j=0}^{\lfloor t/T\rfloor} \left(y(0) 
+ \frac{b (t-j T )}{j+1}\right)  (-1)^j  a^j \frac{(t-j T )^j}{\Gamma (j+1)} \theta (t-j T).
\end{split}\]
    Furthermore, if we divide the time domain into intervals of length $T$, for each $t$, there exists an $n\in\mathbb{N}_0$, such that $t\in [nT,(n+1)T)$ and  $\lfloor t/T\rfloor=n$. 
\end{Remark}
\begin{Proposition}\label{Prop_1}
    For each $t>0$, the solution of \eqref{equation(4)} is 
    \[y(t)=\sum_{j=0}^{\lfloor t/T\rfloor} \left(y(0) 
+ \frac{b (t-j T )}{j+1}\right)  (-1)^j  a^j \frac{(t-j T )^j}{\Gamma (j+1)} \theta (t-j T),\]
with   \[y(t)=y(0) 
+  b t, \quad \text{for} \; t\in [0, T).\] 
\end{Proposition}
\begin{proof} Proposition \ref{Prop_1} is proven directly by applying Remark \ref{Rem_2}. \end{proof}

\subsection{Fractional Caputo Time-Delayed Differential equation}
Next, consider the fractional Caputo's time-delayed homogeneous differential equation of order \(\alpha,\) \(0 < \alpha \leq 1\) given by:
\begin{equation}
{}^C D_t^{\alpha} y(t) + a y(t-T) = 0, \quad y(0)=0, 
\end{equation}
where \({}^C D_t^{\alpha}\) denotes the Caputo fractional derivative and $a$ is a constant. The characteristic equation is:
\begin{equation}
s^{\alpha} + a e^{-sT} = 0,
\end{equation}
with solutions given by the $W$ function
\[s= {\alpha  W\left(\frac{T  (-a)^{\frac{1}{\alpha }}}{\alpha }\right)}/{T}.\]

Now, for the inhomogeneous equation 
\begin{equation}
D_t^{\alpha} y(t) + a y(t-T) = b, \quad y(0)\neq 0, \quad a, b\; \text{constants}, \quad y(t)=0, \; t<0. \label{equation(11)}
\end{equation}
Using the Laplace transform, we obtain:
\begin{equation}
\begin{split}
& s^{\alpha} Y(s) - s^{\alpha-1} y(0) + a Y(s) e^{-sT} = \frac{b}{s} \implies
Y(s) = \frac{y(0)}{s \left(1 + \frac{a  e^{-sT}}{s^{\alpha}}\right)} + \frac{b}{s^{\alpha+1} \left(1 + \frac{a  e^{-sT}}{s^\alpha}\right)} \label{eq12}.
\end{split}
\end{equation}

Using Remark \ref{rem1}, from equation \eqref{eq12} we have 
\begin{equation}
\begin{split}
    Y(s) =y(0) \sum_{j=0}^{\infty} (-1)^j \frac{a^j  e^{-s jT}}{s^{\alpha j+1}} + b \sum_{j=0}^{\infty} (-1)^j \frac{a^j  e^{-s jT}}{s^{\alpha(j+1) +1}},
\quad 0<\left|\frac{a  e^{-sT}}{s^{\alpha}}\right|<1.
\end{split}
\end{equation}
To find the inverse Laplace transform, yielding
\begin{equation}
\begin{split}
       y(t) & =y(0) \sum_{j=0}^{\infty} (-1)^j  a^j\mathcal{L}^{-1}\left[\frac{e^{-s j T}}{s^{\alpha j+1}}\right] + b \sum_{j=0}^{\infty} (-1)^j  a^j\mathcal{L}^{-1}\left[\frac{e^{-s j T}}{s^{\alpha(j+1) +1}}\right]\\
      & = y(0) \sum_{j=0}^{\infty} (-1)^j  a^j \frac{\theta (t-j T ) (t-j T )^{\alpha  j}}{\Gamma (j \alpha +1)}
     +   b \sum_{j=0}^{\infty} (-1)^j  a^j \frac{\theta (t-j T ) (t-j T )^{\alpha(1+j)}}{\Gamma ( (j+1)\alpha +1)}. \label{equation(14)}
\end{split}
\end{equation}
\begin{Remark}\label{Rem_3}
    For each $t>0$, the series in the expression \eqref{equation(14)} is a finite sum. To see this, note that $\theta(t-kT)=0$ for all $k>t/T$. Then, 
    \[
\begin{split}
y(t) &  = \sum_{j=0}^{\lfloor t/T\rfloor} \left(\frac{y(0)}{\Gamma (j \alpha +1)}  + \frac{b (t-j T )^{\alpha}}{\Gamma ( (j+1)\alpha +1)}\right) (-1)^j  a^j  {\theta (t-j T ) (t-j T )^{\alpha  j}}.
\end{split}\]
    Furthermore, if we divide the time domain into intervals of length $T$, for each $t$, there exists an $n\in\mathbb{N}_0$, such that $t\in [nT,(n+1)T)$ and  $\lfloor t/T\rfloor=n$. 
\end{Remark}
\begin{Proposition}\label{Prop_2}
    For each $t>0$, the solution of \eqref{equation(11)} is 
    \[y(t)=\sum_{j=0}^{\lfloor t/T\rfloor} \left(\frac{y(0)}{\Gamma (j \alpha +1)}  + \frac{b (t-j T )^{\alpha}}{\Gamma ( (j+1)\alpha +1)}\right) (-1)^j  a^j  {\theta (t-j T ) (t-j T )^{\alpha  j}},\]
with   \[y(t)=y(0) 
+   \frac{b{t}^{\alpha}}{\Gamma(\alpha +1)},\quad \text{for} \; t\in [0, T).\] 

For $\alpha=1$, we recover the case of derivative or order $1$.
\end{Proposition}
\begin{proof} Proposition \ref{Prop_2} is proven directly by applying Remark \ref{Rem_3}. \end{proof}

\subsection{Higher-Order Fractional Differential Equation with Time Delays}
Finally, let us consider a higher-order fractional differential equation with time delays of order \(\beta,\) \(1 < \beta \leq 2\) given by:
\begin{equation}
{}^C D_t^{\beta} y(t) + a y(t-T) = b, \quad a, b \; \text{constants}, \quad y(0), \; y'(0)\; \text{given}, \quad  y(t)=0, \; t<0, \label{equation(15)}
\end{equation}
where \(D_t^{\beta}\) denotes the Caputo fractional derivative of order \(\beta\). 
Using the Laplace transform:
\begin{equation}
\begin{split}
& \left(s^{\beta} + a e^{-sT}\right) Y(s) - s^{\beta-1} y(0) - s^{\beta-2} y'(0)  = \frac{b}{s}
\implies  Y(s)  = \frac{s^{\beta-1} y(0) + s^{\beta-2} y'(0) + \frac{b}{s}}{s^{\beta} + a e^{-sT}}.
\end{split}
\end{equation}
Using similar arguments as before (Remark \ref{rem1}), we have
 
\begin{equation}
\begin{split}
Y(s)  =y(0) \sum_{j=0}^{\infty} (-1)^j \frac{a^j  e^{-s jT}}{s^{\beta j+1}}+ y'(0) \sum_{j=0}^{\infty} (-1)^j \frac{a^j  e^{-s jT}}{s^{\beta j+2}}  + b \sum_{j=0}^{\infty} (-1)^j \frac{a^j  e^{-s jT}}{s^{\beta(j+1) +1}},
\quad 0 < \left|\frac{a  e^{-sT}}{s^{\beta}}\right|<1.
\end{split}
\end{equation}

To find the inverse Laplace transform, yielding
 
\begin{equation}
\begin{split}
       y(t) & =y(0) \sum_{j=0}^{\infty} (-1)^j  a^j\mathcal{L}^{-1}\left[\frac{e^{-s j T}}{s^{\beta j+1}}\right] + y'(0) \sum_{j=0}^{\infty} (-1)^j  a^j\mathcal{L}^{-1}\left[\frac{e^{-s j T}}{s^{\beta j+2}}\right] + b \sum_{j=0}^{\infty} (-1)^j  a^j\mathcal{L}^{-1}\left[\frac{e^{-s j T}}{s^{\beta(j+1) +1}}\right]\\
     & =  \sum_{j=0}^{\infty} \left[\frac{y(0)}{\Gamma (j \beta +1)}
    + \frac{  y'(0)  (t-j T )}{\Gamma (j \beta +2)} + \frac{b  (t-j T )^{\beta}}{\Gamma ( (j+1)\beta +1)}\right] (-1)^j  a^j 
\theta (t-j T ) (t-j T )^{\beta  j}. \label{equation(18)}
\end{split}
\end{equation}

\begin{Remark}\label{Rem_4}
    For each $t>0$, the series in the expression \eqref{equation(18)} is a finite sum. To see this, note that $\theta(t-kT)=0$ for all $k>t/T$. Then, 
    \[
\begin{split}
y(t) &  = \sum_{j=0}^{\lfloor t/T\rfloor} \left[\frac{y(0)}{\Gamma (j \beta +1)}
    + \frac{  y'(0)  (t-j T )}{\Gamma (j \beta +2)} + \frac{b  (t-j T )^{\beta}}{\Gamma ( (j+1)\beta +1)}\right] (-1)^j  a^j 
\theta (t-j T ) (t-j T )^{\beta  j}.
\end{split}\]
    Furthermore, if we divide the time domain into intervals of length $T$, for each $t$, there exists an $n\in\mathbb{N}_0$, such that $t\in [nT,(n+1)T)$ and  $\lfloor t/T\rfloor=n$. 
\end{Remark}
\begin{Proposition}\label{Prop_3}
    For each $t>0$, the solution of \eqref{equation(15)} is 
    \[y(t)= \sum_{j=0}^{\lfloor t/T\rfloor} \left[\frac{y(0)}{\Gamma (j \beta +1)}
    + \frac{  y'(0)  (t-j T )}{\Gamma (j \beta +2)} + \frac{b  (t-j T )^{\beta}}{\Gamma ( (j+1)\beta +1)}\right] (-1)^j  a^j 
\theta (t-j T ) (t-j T )^{\beta  j},\]
with   \[y(t)= y(0) 
    +   y'(0) t  + \frac{b  t^{\beta}}{\Gamma (\beta +1)},\quad \text{for} \; t\in [0, T).\] 
\end{Proposition}

\begin{proof} Proposition \ref{Prop_3} is proven directly by applying Remark \ref{Rem_4}. \end{proof}

In these examples, we have explored some fractional time-delayed differential equations. We have also discussed first-order and fractional Caputo derivatives, FTDDEs, and higher-order fractional differential equations with time delays. We derived their characteristic equations and solved them using the Laplace transform. These techniques are valuable tools for analyzing and solving complex differential equations with time delays, enhancing our understanding of real-world phenomena.

\section{Time-delayed bulk viscosity}\label{sect.3}

The FLRW cosmological model is widely used because it assumes the universe is homogeneous and isotropic on large scales, making perturbation studies more manageable. Its simplicity and effectiveness in describing large-scale cosmic evolution make it the standard framework.  

Recent extensions, such as the Dipolar Cosmological Principle, incorporate cosmic flows and axial anisotropies \cite{Krishnan:2022qbv, Krishnan:2022uar}, broadening FLRW’s applicability. This approach enables modeling complex dynamics using more generalized differential equations, including anisotropic shear and expansion. Further research has explored coupled scalar fields and oscillatory behaviors in specific metrics, offering insights into dark energy \cite{Orjuela-Quintana:2021zoe, Allahyari:2023kfm}.  

As a first approach, we assume that the flat FLRW metric will serve as the foundation for this study:  
\begin{equation}  
    ds^2=-dt^2+a^2(t)\left(dx^2+dy^2+dz^2\right),  
\end{equation}  
along with the energy-momentum tensor incorporating a bulk viscosity term:  
\begin{equation}  
    T_{\mu\nu}=\rho u_\mu u_\nu+(p+\eta)h_{\mu\nu},  
\end{equation}  
where \( u^\mu=\delta_0^\mu \) represents the four-velocity of the comoving observer, while \( h_{\mu\nu}=g_{\mu\nu}+u_\mu u_\nu \) is the projective tensor. Here, \( \rho \) and \( p \) denote the energy density and pressure of the perfect fluid.  

The function \( \eta = \eta(\rho) \) represents the bulk viscosity term, as discussed in \cite{Murphy:1973zz, Barrow:1986yf, padmanabhan1987viscous, Gron:1990ew, Zimdahl:1996ka, Szydlowski:2006ma}. This term appears in the spatial component of space-time:  
\begin{equation}  
    T_{\mu\nu}=\left(\rho+p\right)u_\mu u_\nu +pg_{\mu\nu}+\eta h_{\mu\nu}.  
\end{equation}  

The introduction of the bulk viscosity term modifies the Friedmann equations:  
\begin{equation} \label{Feqs}  
    3H^2=\rho,  
\end{equation}  
\begin{equation} \label{Feqs_2}  
    2\dot{H}+3H^2+p-\eta=0.  
\end{equation}  
While the continuity equation for the perfect fluid reads:  
\begin{equation}  
    \dot{\rho}+3H\left(\rho+p\right)=3H\eta.  
\end{equation}  
When the perfect fluid behaves as an ideal gas, i.e., \( p = (\gamma -1) \rho \), and substituting into \eqref{Feqs}, equation \eqref{Feqs_2} simplifies to:  
\begin{equation}  
    2\dot{H}+3\gamma H^2-\eta=0. \label{eq:bulk_viscosity}  
\end{equation}  

The function \( \eta = \eta(\rho) \) can take various forms, enabling descriptions of alternative cosmological models, such as the Chaplygin gas and its modifications \cite{Kamenshchik:2001cp, Bento:2002ps, Zhu:2004aq, Xu:2012qx}. The Chaplygin gas model is recovered when the perfect fluid behaves as an ideal gas and \( \eta = \eta_0 \rho^{-\lambda} \).  
Originally proposed as a unified dark matter model, the Chaplygin gas is also relevant for early-universe scenarios, particularly inflation \cite{Amendola:2003bz, Barrow:1988yc, Barrow:1990vx, Villanueva:2015ypa}. Several Chaplygin gas-like cosmologies can be incorporated within this framework \cite{Barrow:1990vx, Villanueva:2015ypa, Barrow:2016qkh, Israel:1979wp}.  

Additionally, the right-hand side of the conservation law introduces particle creation and destruction, which play a crucial role in different phases of cosmic evolution. Particle production processes have essential applications in both the early and late universe \cite{Prigogine:1989zz, Abramo:1996ip, steigman2009accelerating, lima2010cosmol}, as well as in extended frameworks \cite{Basilakos:2010yp, Jesus:2011ek, Mimoso:2013zhp, Lima:2014qpa}.  

Eckart’s theory serves as a first approximation of bulk viscosity models and was later refined by the Israel-Stewart formalism \cite{Cruz:2016rqi, Disconzi:2014oda, Leyva:2016gzm} and \cite{Acquaviva:2018rqi, Aguilar-Perez:2022bzb}. The Israel-Stewart model introduces additional degrees of freedom in the field equations, offering a more comprehensive description of physical variables and addressing limitations inherent in Eckart’s formulation like non-causality.

According to \cite{Paliathanasis:2022wwl}, the simplest bulk viscosity scenario in \eqref{eq:bulk_viscosity}, following Eckart’s formulation, arises when the viscosity term depends on \( H \), say \( \eta(H) \).  
Provided that \( \eta(H) - 3\gamma H^{2} \neq 0 \), equation \eqref{eq:bulk_viscosity} can be solved explicitly via quadratures:
\begin{equation}
\int \frac{2 dH}{\eta(H) - 3\gamma H^{2}} = t - t_{0}. \label{eq:quadrature_solution}
\end{equation}  
For specific forms of bulk viscosity, equation \eqref{eq:bulk_viscosity} simplifies to well-known differential equations with closed-form solutions. When \( \eta(H) \) is linear, it takes the form of a Riccati first-order ODE, while a third-order polynomial \( \eta(H) \) results in an Abel equation \cite{Paliathanasis:2022wwl}.

Let \( H_0 \) be a zero of the function \( f(H) \), defined as:
\begin{equation}
f(H) = \frac{\eta(H) - 3\gamma H^{2}}{2}. \label{eq:critical_point}
\end{equation}  
From a physical perspective, when \( H_0 \neq 0 \), the critical point describes a de Sitter universe, whereas for \( H_0 = 0 \), the resulting spacetime is an empty Minkowski space. Since \( \eta(H) \) is a real function, periodic behavior near the critical point is not expected \cite{Paliathanasis:2022wwl}.

In \cite{Paliathanasis:2022wwl}, the most straightforward extension of the bulk viscosity scenario in \eqref{eq:bulk_viscosity} is the introduction of a time delay in the \( H \)-function within the field equations.

Scientists argue that vacuum energy density is unlikely to remain static in an expanding universe, prompting the exploration of a smooth time-dependent vacuum energy \cite{Basilakos:2018xjp, Cruz:2023dzn, Fritzsch:2016ewd, Sola:2015rra}. One approach models the cosmological constant as a decreasing function, addressing the Hubble constant tension \cite{amendola2015dark, DiValentino:2017gzb} and aligning theoretical predictions with observed values based on Quantum Field Theory in curved classical spacetime \cite{parker2009quantum}.

In \cite{El-nabulsi:2016dsj}, a generalized fractional scale factor and a time-dependent Hubble parameter are introduced, governed by an Ornstein-Uhlenbeck-like fractional differential equation. This model describes the accelerated expansion of a non-singular universe, both with and without scalar fields, revealing previously hidden cosmological features \cite{El-nabulsi:2016dsj}, inspired by the stochastic formulation proposed in \cite{10.1143/PTPS.139.470}.

An alternative approach to modeling bulk viscosity is provided by the Israel-Stewart formalism, with its simplest case governed by \cite{Maartens:1996vi, Murphy:1973zz, Barrow:1986yf, Szydlowski:2006ma}:
\begin{equation}
\tau \dot{\eta} + \eta = 3\xi H. \label{eq:bulk_viscosity_simple}
\end{equation}  
where \( \tau \) is the relaxation time, given by \( \tau = \xi \rho^{-1} \), with \( \xi \) as the bulk viscosity coefficient. When \( \xi = 3\gamma \xi_0 \rho^{\kappa} \), the system admits a real critical point corresponding to the de Sitter solution, \( H_P = \xi_0^{\frac{1}{1-\kappa}} \). The linearized system near the critical point exhibits imaginary eigenvalues when \( \kappa < -1 - \frac{1}{3\gamma^{2}} - 3\gamma^{2} \), implying spiral behavior. However, for \( \kappa > 0 \) \cite{Murphy:1973zz, Barrow:1986yf, Szydlowski:2006ma, Gron:1990ew}, oscillatory behavior is absent, unlike the time-delay model, which supports oscillations for positive \( \kappa \) \cite{Paliathanasis:2022wwl}.

In the full Israel-Stewart framework, the bulk viscous pressure \( \eta \) obeys a causal evolution equation with second-order corrections in deviations from equilibrium. Specifically, the equation governing \( \eta \) takes the form:
\begin{equation}
\tau \dot{\eta} + \eta + \frac{1}{2} \left( 3H + \frac{\dot{\tau}}{\tau} - \frac{\dot{\xi}}{\xi} - \frac{\dot{\mathcal{T}}}{\mathcal{T}} \right) \tau \eta = \xi H, \label{eq:bulk_viscosity_full}
\end{equation}  
where \( \tau \) is the relaxation time and \( \mathcal{T} \) is the barotropic temperature of the viscous fluid. 
See Ref. \cite{Maartens:1996vi} for a detailed derivation of the transport equations.
For barotropic fluids with a constant barotropic index, \( \gamma_v=1+w \), the relaxation time can be reduced to:
\begin{equation}
\tau = \frac{\xi}{(2-\gamma_v)\gamma_v\rho_v}.
\end{equation}  
At the same time, the Gibbs integrability condition allows us to calculate the temperature \( \mathcal{T} \) as:
\begin{equation}
\mathcal{T}\propto \rho_v^{\frac{\gamma_v-1}{\gamma_v}}.
\end{equation}  
By adjusting the functions \( \tau \) and \( \xi \), various cosmological scenarios can be recovered \cite{Murphy:1973zz, Barrow:1986yf, Maartens:1996vi} and \cite{Szydlowski:2006ma, Leyva:2017arj}.

\subsection{Time-delayed bulk viscosity}
Despite the strengths of FLRW cosmology, refinements are needed to address deviations from perfect fluid behavior. Viscosity and retardation effects regulate shear dynamics, impose causality constraints, and influence anisotropic expansion, gravitational wave propagation, and energy dissipation.

Time delays in cosmic evolution introduce nonlocality, meaning present states depend on past interactions rather than evolving instantaneously. Fractional calculus provides a framework for modeling these effects, particularly anomalous transport phenomena in astrophysical systems. By incorporating nonlocal derivatives, fractional models capture cumulative influences from previous events, offering insights into irregular matter distributions and cosmic energy exchanges \cite{uchaikin2017fractional, herrmann2014fractional}.  

The framework in \cite{herrmann2014fractional} unifies damping and shear effects within fractional dynamics, aiding the modeling of dissipative processes in cosmology. Meanwhile, \cite{uchaikin2017fractional} describes stochastic jumps and nonlocal transitions from a statistical mechanics perspective. Integrating these approaches enhances models of transport mechanisms and astrophysical irregularities, improving the understanding of entropy generation and turbulence.  

Time delays in cosmological models account for non-instantaneous interactions between cosmic components, such as the lag in energy exchange between matter and radiation. This is particularly relevant in the early universe, where viscosity may arise from quantum or thermal processes.  
In this context, the viscosity function is given by  
\begin{equation}\label{eta}
    \eta(t) = 2\eta_0 H(t-T).
\end{equation}  
Following \cite{Paliathanasis:2022wwl}, this expression arises from a modified cosmological framework. Our approach extends \cite{Paliathanasis:2022wwl} by incorporating fractional derivatives and time-delay corrections to evolution equations. A rigorous derivation from Einstein’s field equations would clarify the role of viscosity and retardation in cosmic dynamics. Applying Laplace transformations requires justifying assumptions and examining their implications for nonlocal transport in cosmology. This phenomenological approach seeks to determine whether the modified framework recovers standard viscous or non-viscous cosmology in specific limits.

The corresponding evolution equation is  
\begin{equation}\label{eqlinearviscosity}
    \dot{H}(t) + \frac{3\gamma}{2} H^2(t) - \eta_0 H(t-T) = 0.
\end{equation}  
Taking the limit \(T \to 0\), this equation simplifies to the Riccati equation  
\begin{equation}\label{Riccati}
    \dot{H}=f(H) := - H\left(\frac{3\gamma}{2}H - \eta_0\right).
\end{equation}  
The solution is  
\begin{equation}
    H(t) = \frac{2 \eta_0 H_0 e^{\eta_0 t}}{2 \eta_0+3 \gamma H_0 \left(e^{\eta_0 t}-1\right)}, \quad \frac{\dot{a}}{a}=H \implies a(t) = a_0 \left(\frac{3 \gamma H_0\left(e^{\eta_0 t}-1\right)}{2 \eta_0}+1\right)^{\frac{2}{3\gamma}}.
\end{equation}  
The critical points of equation \eqref{Riccati}, where \( f(H) = 0 \), are \( H_A = 0 \) and \( H_B = (2\eta_0)/(3\gamma) \).  

For \( H_A \), referring to equation \eqref{Feqs}, this corresponds to a universe with zero energy density (\(\rho = 0\)). This condition describes an empty universe or a universe where \(a \to \infty\), with vanishing matter and energy density.  

For \( H_B \), we obtain  
\begin{equation}
H_B = \frac{\dot{a}}{a} = H_B \implies a(t) = a_0 e^{H_Bt}.
\end{equation}  
This exponential expansion characterizes the de Sitter phase.  

The stability of the critical point is determined by the sign of \( \frac{df}{dH} |_{H=H_{0}} \). If \( \frac{df}{dH} |_{H=H_{0}} < 0 \), the critical point is an attractor, while if \( \frac{df}{dH} |_{H=H_{0}} > 0 \), the critical point is unstable.  

Applying the variable transformation \( y(t) = H(t) - H_B \), equation \eqref{eqlinearviscosity} transforms into  
\begin{equation}\label{nonfractandnonlinear}
    \dot{y}(t) + \frac{3\gamma}{2} y^2(t) + 2\eta_0 y(t) - \eta_0 y(t-T) = 0.
\end{equation}  
\subsection{Linearization}
\label{SECT:3.1}
We linearized the last equation around $y(t)=0$, to obtain:
\begin{equation}\label{nonfractionalandretarded}
    \dot{y}(t)+2\eta_0y(t)-\eta_0y(t-T)=0.
\end{equation}
Now applying the Laplace transform to equation \eqref{nonfractionalandretarded}, we obtain
\begin{equation}
    s\mathcal{L}\left\{y(t)\right\}-y(0)+2\eta_0 \mathcal{L}\left\{y(t)\right\}-\eta_0e^{-sT}\mathcal{L}\left\{y(t)\right\}=0.
\end{equation}
Combining the steps, we get:
\begin{align}
\mathcal{L}\{y(t)\} &=y(0)\sum_{j=0}^\infty \frac{\left(\eta_0e^{-sT}-2\eta_0\right)^j}{s^{j+1}}\; \text{for} \; 0<\left|\eta_0 s^{-1}\left(e^{-sT}-2\right)\right|<1.
\end{align}
Using the Newton binomial
we have
\begin{equation}
\begin{split}
    \left(\eta_0e^{-sT}-2\eta_0\right)^j&=\sum_{k=0}^j\frac{j!}{k!(j-k)!}(-2)^{j-k}\eta_0^je^{-skT}.
\end{split}
\end{equation}
Thus,
\begin{equation}
\begin{split}
    \mathcal{L}\left\{y(t)\right\}&=y(0)\sum_{j=0}^\infty \sum_{k=0}^j\frac{j!}{k!(j-k)!}(-2)^{j-k}\eta_0^j\frac{e^{-skT}}{s^{j+1}}.
\end{split}
\end{equation}
Applying the inverse Laplace transform
\begin{equation}
    y(t)=y(0)\sum_{j=0}^\infty \sum_{k=0}^j\frac{j!}{k!(j-k)!}(-2)^{j-k}\eta_0^j\mathcal{L}^{-1}\left[\frac{e^{-skT}}{s^{j+1}}\right],
\end{equation}
we obtain
\begin{equation}
\begin{split}
    y(t)&=y(0)\sum_{j=0}^\infty \sum_{k=0}^j\frac{j!}{k!(j-k)!}(-2)^{j-k}\eta_0^j\frac{(t-kT)^j\theta(t-kT)}{\Gamma(j+1)}\\
    &=y(0)\sum_{k=0}^\infty \sum_{j=k}^\infty\frac{j!}{k!(j-k)!}(-2)^{j-k}\eta_0^j\frac{(t-kT)^j\theta(t-kT)}{\Gamma(j+1)}\\
    &=y(0)\sum_{k=0}^\infty \sum_{j=k}^\infty\frac{\left[-2\eta_0(t-kT)\right]^j}{\Gamma(j+1)}(-2)^{-k}\frac{j!}{k!(j-k)!}\theta(t-kT)\\
    &=y(0)\Bigg\{\sum_{j=0}^\infty\frac{(-2\eta_0t)^j}{\Gamma(j+1)}\theta(t)+\sum_{j=1}^\infty \frac{\left[-2\eta_0(t-T)\right]^j}{\Gamma(j+1)}(-2)^{-1}\frac{j!}{1!(j-1)!}\theta(t-T)\\ & \quad \quad \quad \quad \quad +\sum_{j=2}^\infty \frac{\left[-2\eta_0(t-2T)\right]^j}{\Gamma(j+1)}(-2)^{-2}\frac{j!}{2!(j-2)!}\theta(t-2T)+\cdots\Bigg\}\\
    &=y(0)\Bigg\{e^{-2\eta_0t}\theta(t)+(-2)^{-1}\frac{\theta(t-T)}{1!}\left[-2\eta_0(t-T)\right]e^{-2\eta_0(t-T)}\\ & \quad \quad \quad \quad +(-2)^{-2}\frac{\theta(t-2T)}{2!}\left[-2\eta_0(t-2T)\right]^2 e^{-2\eta_0(t-2T)}+\cdots\Bigg\} \\
    &=y(0)\sum_{k=0}^\infty\frac{e^{-2\eta_0(t-kT)}\left[\eta_0(t-kT)\right]^k\theta(t-kT)}{\Gamma(k+1)}.
    \end{split}
\end{equation}
Therefore, the analytical solution of equation \eqref{nonfractionalandretarded} is given by
\begin{equation}\label{Analyticaly1}
    y(t)=y(0)\sum_{k=0}^\infty\frac{e^{-2\eta_0(t-kT)}\left[\eta_0(t-kT)\right]^k\theta(t-kT)}{\Gamma(k+1)}.
\end{equation}
Hence, given $y(0)=H_0-H_B$, $H(t)$ is
\begin{equation}\label{AnalyticalH1}
    H(t)=H_B+\left(H_0-H_B\right)\sum_{k=0}^\infty\frac{e^{-2\eta_0(t-kT)}\left[\eta_0(t-kT)\right]^k\theta(t-kT)}{\Gamma(k+1)},
\end{equation}
where $H_0=H(t=0)$, and from 
\begin{align}
    a(t) &= \exp\left[\int H(t) dt\right]. \label{Togeta(t)}
    \end{align}
    Hence, omitting a multiplicative factor that we set to $1$, 
    \begin{align}
     a(t) &=\exp\int\left[ H_B+\left(H_0-H_B\right)\sum_{k=0}^\infty\frac{e^{-2\eta_0(t-kT)}\left[\eta_0(t-kT)\right]^k\theta(t-kT)}{\Gamma(k+1)}\right] dt,
     \\
     &= e^{ H_B t} \prod_{k=0}^\infty \frac{\exp\left(\frac{2^{-(k+1)} \left(H_0-H_B\right) \theta (t-k T)}{\eta_{0}}\right)}{\exp \left(\frac{2^{-(k+1)} \left(H_0-H_B\right) \theta (t-k T) \Gamma
   (k+1,2 \eta_0 (t-k T))}{\eta_0 \Gamma (k+1)}\right)},
   \label{Scale-factor-a}
\end{align}
where
\(\Gamma (a,z)\) is the incomplete Gamma function 
\begin{equation}
    \Gamma (a,z)=\int_{z}^{\infty}t^{a-1}e^{-t}dt, \; \text{where}\; 
    \Gamma (a,0)=\Gamma (a). \label{1.1.1.8}
\end{equation}
\begin{Remark}\label{Rem_5}
    For each $t>0$, the series in \eqref{Analyticaly1} is a finite sum. To see this, note that $\theta(t-kT)=0$ for all $k>t/T$. Then, 
    \[y(t)=\left(H_0-H_B\right)\sum_{k=0}^{\lfloor t/T\rfloor}\frac{e^{-2\eta_0(t-kT)}\left[\eta_0(t-kT)\right]^k \theta(t-kT)}{\Gamma(k+1)}.\]
    Furthermore, if we divide the time domain into intervals of length $T$, for each $t$, there exists an $n\in\mathbb{N}_0$, such that $t\in [nT,(n+1)T)$ and  $\lfloor t/T\rfloor=n$. 
\end{Remark}
\begin{Proposition}\label{Prop_4}
 For each $t>0$, the solution of \eqref{nonfractionalandretarded} is 
  \[y(t)=y(0)\sum_{k=0}^{\lfloor t/T\rfloor}\frac{e^{-2\eta_0(t-kT)}\left[\eta_0(t-kT)\right]^k \theta(t-kT)}{\Gamma(k+1)},\]
and
\begin{equation}
y(t)=  y(0)e^{-2\eta_0 t}, \quad \text{for} \; t\in [0, T). 
\end{equation}
\end{Proposition}
\begin{proof} Proposition \ref{Prop_4} is proven directly by applying Remark \ref{Rem_5}. \end{proof}
\begin{Remark}\label{Rem_6}
    For each $t>0$, the product in \eqref{Scale-factor-a} is finite. To see this, note that $\theta(t-kT)=0$ for all $k>t/T$. Then, 
    
    \[\prod_{k=0}^\infty \frac{\exp\left(\frac{2^{-(k+1)} \left(H_0-H_B\right) \theta (t-k T)}{\eta_{0}}\right)}{\exp \left(\frac{2^{-(k+1)} \left(H_0-H_B\right) \theta (t-k T) \Gamma
   (k+1,2 \eta_0 (t-k T))}{\eta_0 \Gamma (k+1)}\right)}=\prod_{k=0}^{\lfloor t/T\rfloor} \frac{\exp\left(\frac{2^{-(k+1)} \left(H_0-H_B\right) \theta (t-k T)}{\eta_{0}}\right)}{\exp \left(\frac{2^{-(k+1)} \left(H_0-H_B\right) \theta (t-k T) \Gamma
   (k+1,2 \eta_0 (t-k T))}{\eta_0 \Gamma (k+1)}\right)}.\]
   
    Furthermore, if we divide the time domain into intervals of length $T$, for each $t$, there exists an $n\in\mathbb{N}_0$, such that $t\in [nT,(n+1)T)$, and  $\lfloor t/T\rfloor=n$. 
\end{Remark}
\begin{Proposition}\label{Prop_5}
  For each $t>0$, 
\begin{align}
    H(t)& =H_B+\left(H_0-H_B\right)\sum_{k=0}^{\lfloor t/T\rfloor}\frac{e^{-2\eta_0(t-kT)}\left[\eta_0(t-kT)\right]^k\theta(t-kT)}{\Gamma(k+1)}, \label{AnalyticalHP1}
\\
  a(t)  & = e^{ H_B t} \prod_{k=0}^{\lfloor t/T\rfloor} \frac{\exp\left(\frac{2^{-(k+1)} \left(H_0-H_B\right) \theta (t-k T)}{\eta_{0}}\right)}{\exp \left(\frac{2^{-(k+1)} \left(H_0-H_B\right) \theta (t-k T) \Gamma
   (k+1,2 \eta_0 (t-k T))}{\eta_0 \Gamma (k+1)}\right)}.
\end{align}
\end{Proposition}
\begin{proof} Proposition \ref{Prop_5} is proved by using \eqref{AnalyticalH1}, \eqref{Scale-factor-a} and applying Remarks \ref{Rem_5} and \ref{Rem_6}. \end{proof}
\begin{Proposition}\label{Prop_6}
    As $t\rightarrow +\infty$ the exponential term $e^{-2 \eta_0 t}$ dominates. Therefore, 
    \[\lim_{t\rightarrow +\infty} H(t)  = H_B.\]
\end{Proposition}
\begin{proof}
    Denoting by $S_n(t)$ the $n$-term of the sum in \eqref{AnalyticalH1}, we can see that \[S_{n+1}(t)\leq \eta_0 e^{2\eta_0 T}S_n(t)\frac{t-n T}{n +1} \leq \eta_0 t e^{2\eta_0 T} S_n(t).\]
    Applying this recursively, we get $S_n(t)\leq \eta_0 e^{2\eta_0 T}e^{-2\eta_0 t}$, for $n\geq 1$. The result follows by replacing \eqref{AnalyticalH1} and passing the limit when $t\to +\infty$.
\end{proof}
Moreover, in standard cosmology, we have the deceleration parameter $q$, which tells us whether the Universe's expansion is accelerated or decelerated. It is defined as
\begin{align}
    q(t)= -1-\frac{\dot{H}(t)}{H(t)^2}, \label{weffa}
\end{align}
and the function $w_{\text{eff}}$ represents the behaviour of the fluid, given by
\begin{align}
    w_{\text{eff}}(t)= -1- \frac{2 \dot{H}(t)}{3 H(t)^2}  = \frac{(2 q(t) -1)}{3}\label{weff}.
\end{align}
\begin{Proposition}\label{Prop_7b}
In the initial interval $[0,T)$
we have 
\end{Proposition}
\begin{equation}
      H(t)=H_B+\left(H_0-H_B\right) e^{-2\eta_0 t}, \label{H-ini}
\end{equation}
\begin{equation}
   q(t) = -1-\frac{6 \gamma  \eta_0 e^{2 \eta_0 t} (2 \eta_0-3 \gamma  H_0)}{\left(3 \gamma  H_0+2 \eta_0 \left(e^{2 \eta_0
   t}-1\right)\right)^2},  \label{q-ini}
\end{equation}
\begin{equation}
    w_{\text{eff}}(t) = -1-\frac{4 \gamma  \eta_0 e^{2 \eta_0 t} (2 \eta_0-3 \gamma  H_0)}{\left(3 \gamma  H_0+2 \eta_0 \left(e^{2 \eta_0
   t}-1\right)\right)^2}. \label{w-ini}
\end{equation}
\begin{proof} Proposition \ref{Prop_7b} is proven using \eqref{AnalyticalHP1}, \eqref{weffa} and \eqref{weff}. Using continuity, for $t=0$, 
$H(0) = H_0, \quad q(0) = -1 - \frac{4\eta_0^2}{3\gamma H_0^2} + \frac{2\eta_0}{H_0}
, \quad w_{\text{eff}}(0) = -1 - \frac{8\eta_0^2}{9\gamma H_0^2} + \frac{4\eta_0}{3H_0}
$. \end{proof}
However, using equations \eqref{weffa}, \eqref{weff} which are from the standard model of cosmology, and using \eqref{eqlinearviscosity} to replace $\dot{H}(t)$ we explicitly have the deceleration parameter $q(t)$ and $\omega_{\text{eff}}(t)$, which depends of retarded time $T$ we have for $t>T$:
\begin{align}
  q(t)  & = -1+\frac{3\gamma}{2}- \eta_0 \frac{H(t-T)}{H(t)^2}, \label{q_2} \\
   w_{\text{eff}}(t) & = -1+\gamma- \frac{2}{3}\eta_0 \frac{H(t-T)}{H(t)^2}. \label{weff_2} 
\end{align}
\begin{Proposition}\label{Prop_7}
  As $t\rightarrow +\infty$ the exponential term $e^{-2 \eta_0 t}$ dominates. Therefore
\begin{equation}\lim_{t\rightarrow +\infty} q(t)  = -1, \quad \lim_{t\rightarrow +\infty} \omega_{\text{eff}}(t)  = -1.
\end{equation}
\end{Proposition}
\begin{proof}
    Consequence of definitions \eqref{weffa}, \eqref{weff}  and Proposition \ref{Prop_6}.
\end{proof}
The exponential growth of the scale factor aligns with inflationary cosmology but may also arise naturally from time-delay effects within our framework. While inflation is widely accepted, alternative theories, such as cyclic cosmologies and emergent spacetime models, challenge conventional views \cite{deCesare:2016rsf}. Cyclic cosmological solutions have been extensively studied, though previous research has primarily focused on cyclic universes centered around a static Einstein universe \cite{Ijjas:2019pyf, Clifton:2007tn, Marosek:2015uza, Steinhardt:2001st} and \cite{Barrow:1995cfa, Saridakis:2018fth}.  
The averaging approach has also been applied to determine periodic behaviors in cosmology \cite{Leon:2021lct, Leon:2021rcx, Leon:2021hxc}. In \cite{Ijjas:2019pyf}, a cyclic cosmological model was proposed in which the scale factor undergoes exponential growth in each cycle. This framework addressed various early-universe problems, including the horizon, isotropy, and flatness issues.

\subsection{Error Estimation}
For a given $t > 0$, and based on Remark \ref{Rem_5}, the sum in \eqref{AnalyticalHP1} is effectively finite. Consequently, the solution can be considered exact, ensuring that the analytical representation accurately characterizes the specified range.
In Section \ref{appF1}, we present an optimized algorithm to implement the exact solutions \eqref{AnalyticalHP1}, \eqref{q-ini} together with \eqref{q_2}, and \eqref{w-ini} together with \eqref{weff_2}, as provided in Section \ref{SECT:3.1}.
Figure \ref{Fig00} shows the analytical solution for $H(t)$ from \eqref{AnalyticalHP1}, where $\gamma = 4/3$ corresponds to radiation and $\gamma = 1$ to matter. In both cases, viscosity, defined by the function $\eta$ as a function of the retarded time $T$ in equation \eqref{eta}, drives the Universe's expansion. Recently, the expansion has approached a de Sitter space-time. The model agrees with current observations and supports Proposition \ref{Prop_6}, showing that a scalar field is not the only way to accelerate the Universe's expansion.
\begin{figure}[h]
\centering
\includegraphics[width=0.5\textwidth]{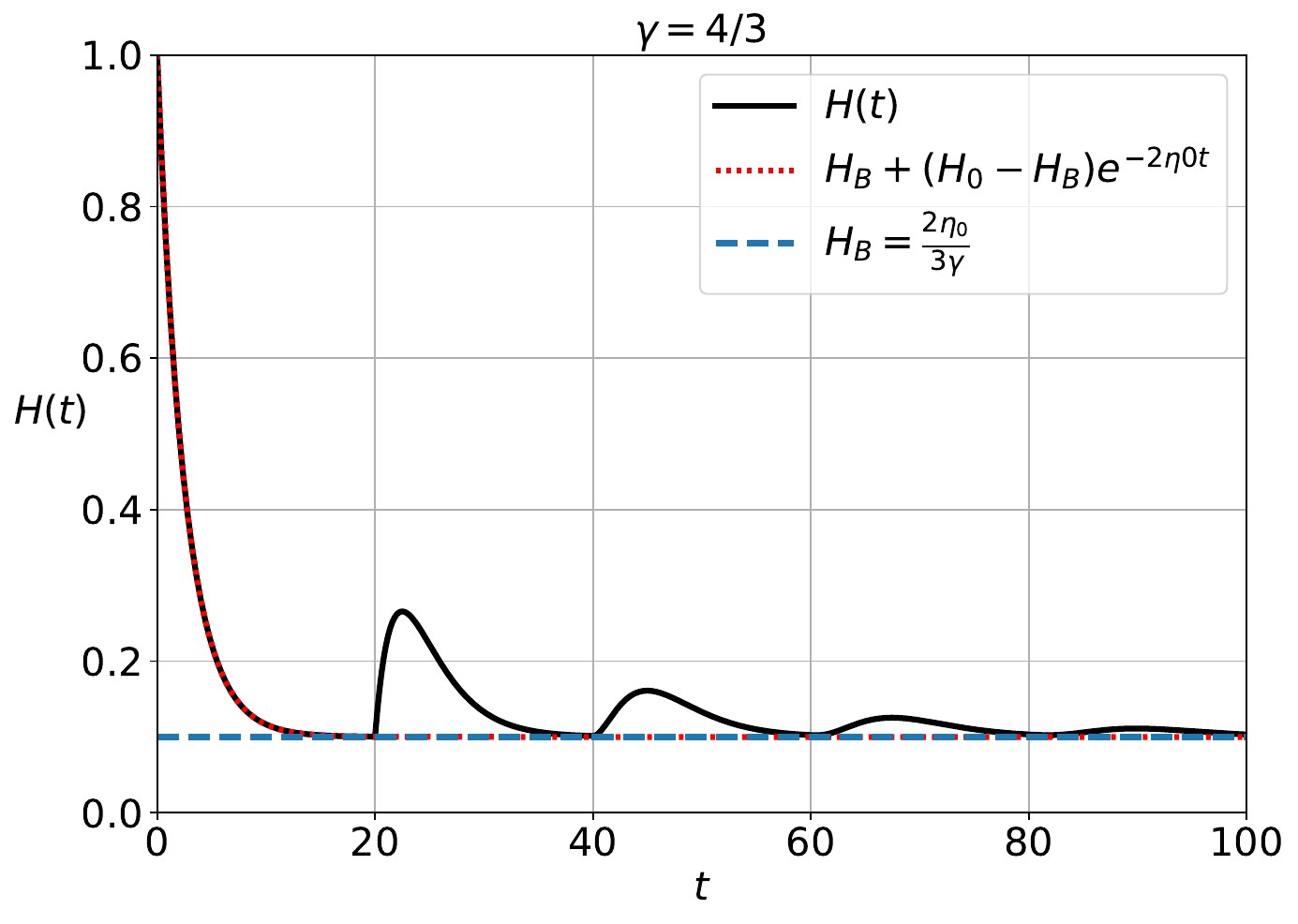}\includegraphics[width=0.5\textwidth]{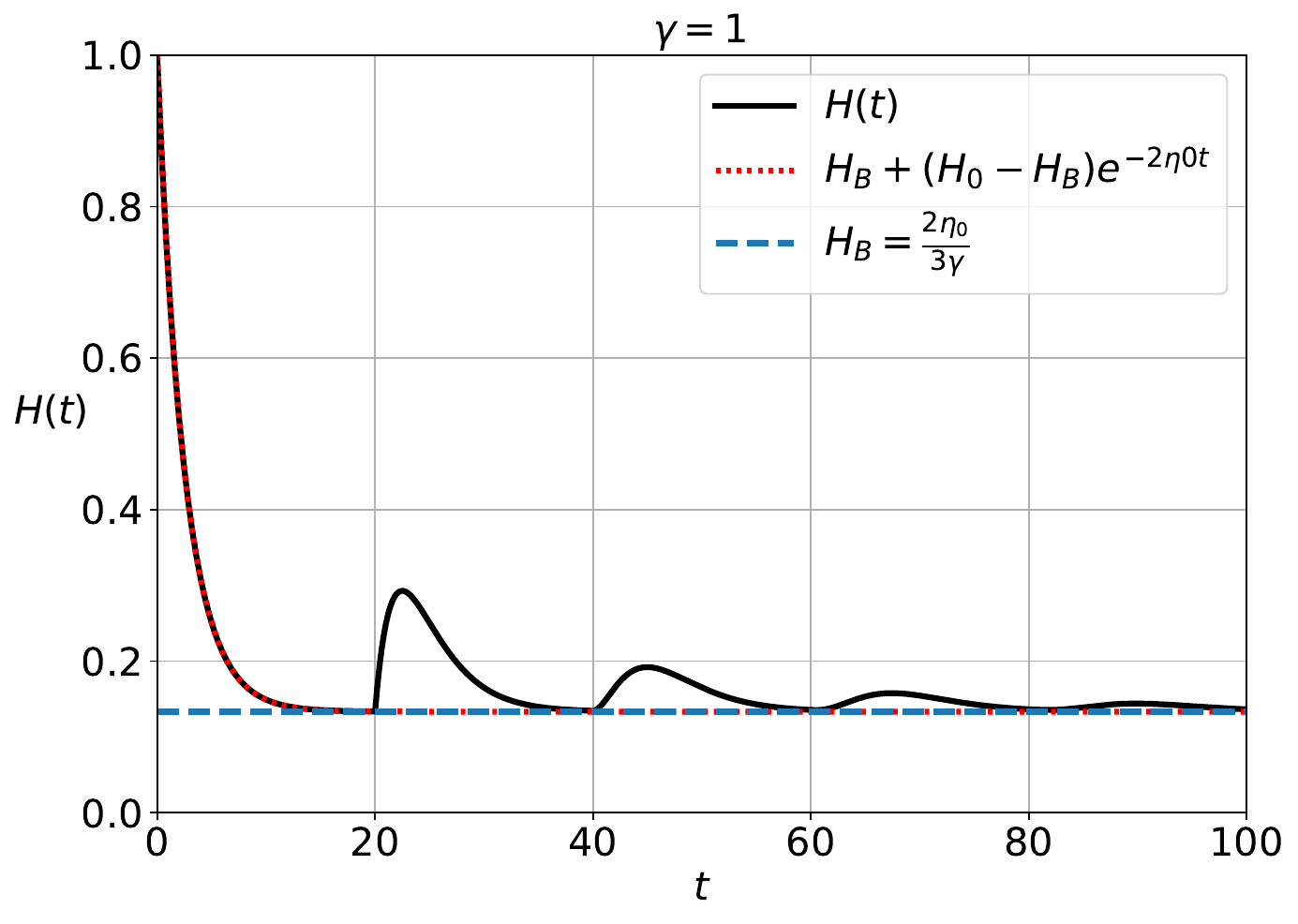}
\caption{\label{Fig00} Analytical solution $H(t)$, for the cases $\gamma=4/3,1$. The other parameters are $\eta_0=0.2$, $T=20$, $H_0=1$, $y_0= H_0-H_B$. The dashed line represents the de Sitter solution.}
\end{figure}
Figure \ref{Fig0ObsB} presents the analytical solutions for $q(t)$ and $\omega_{\text{eff}}(t)$ for $\gamma = 4/3$ and $\gamma = 1$. 
\begin{figure}[h]
\centering
\includegraphics[width=0.5\textwidth]{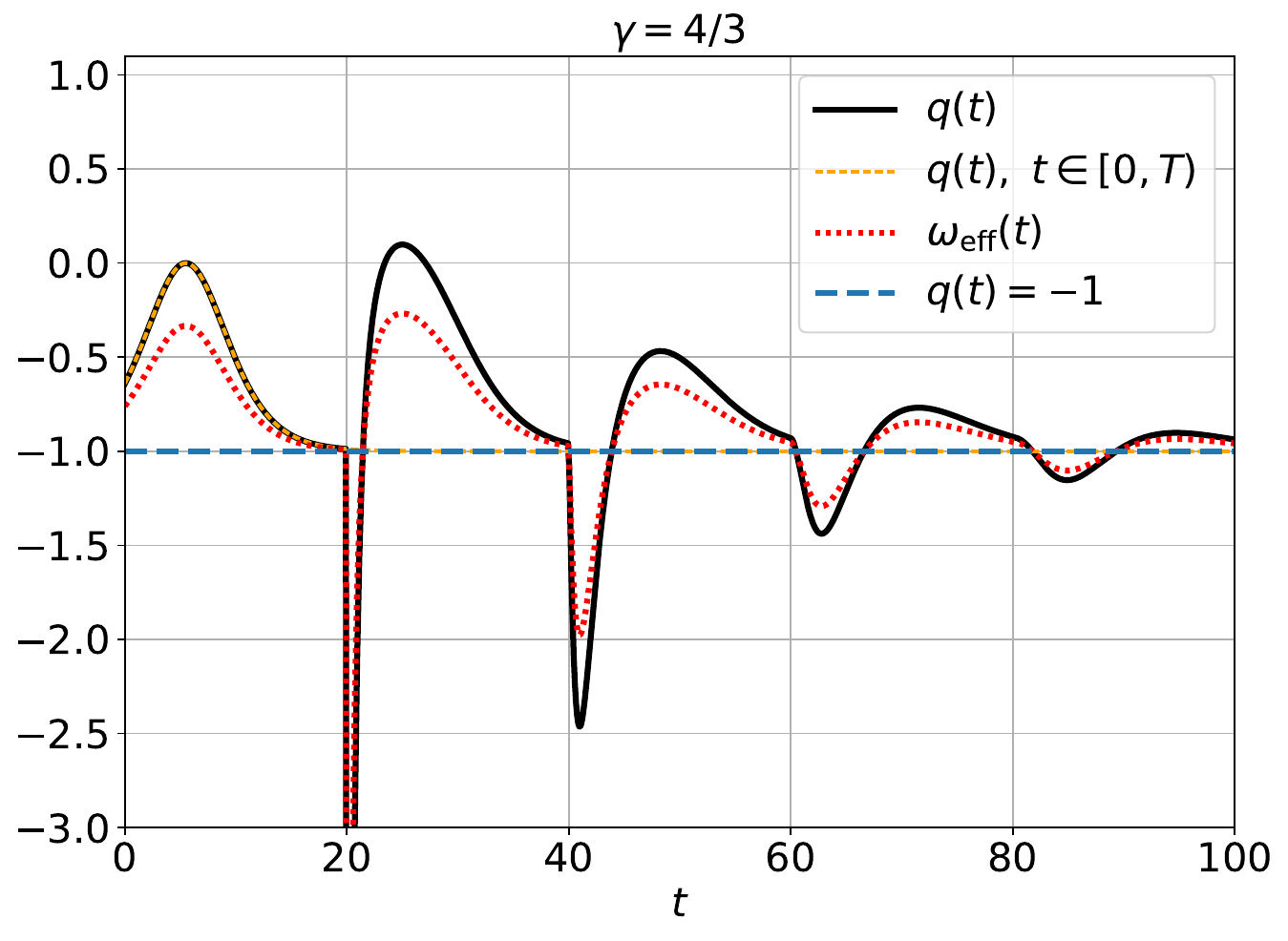}\includegraphics[width=0.5\textwidth]{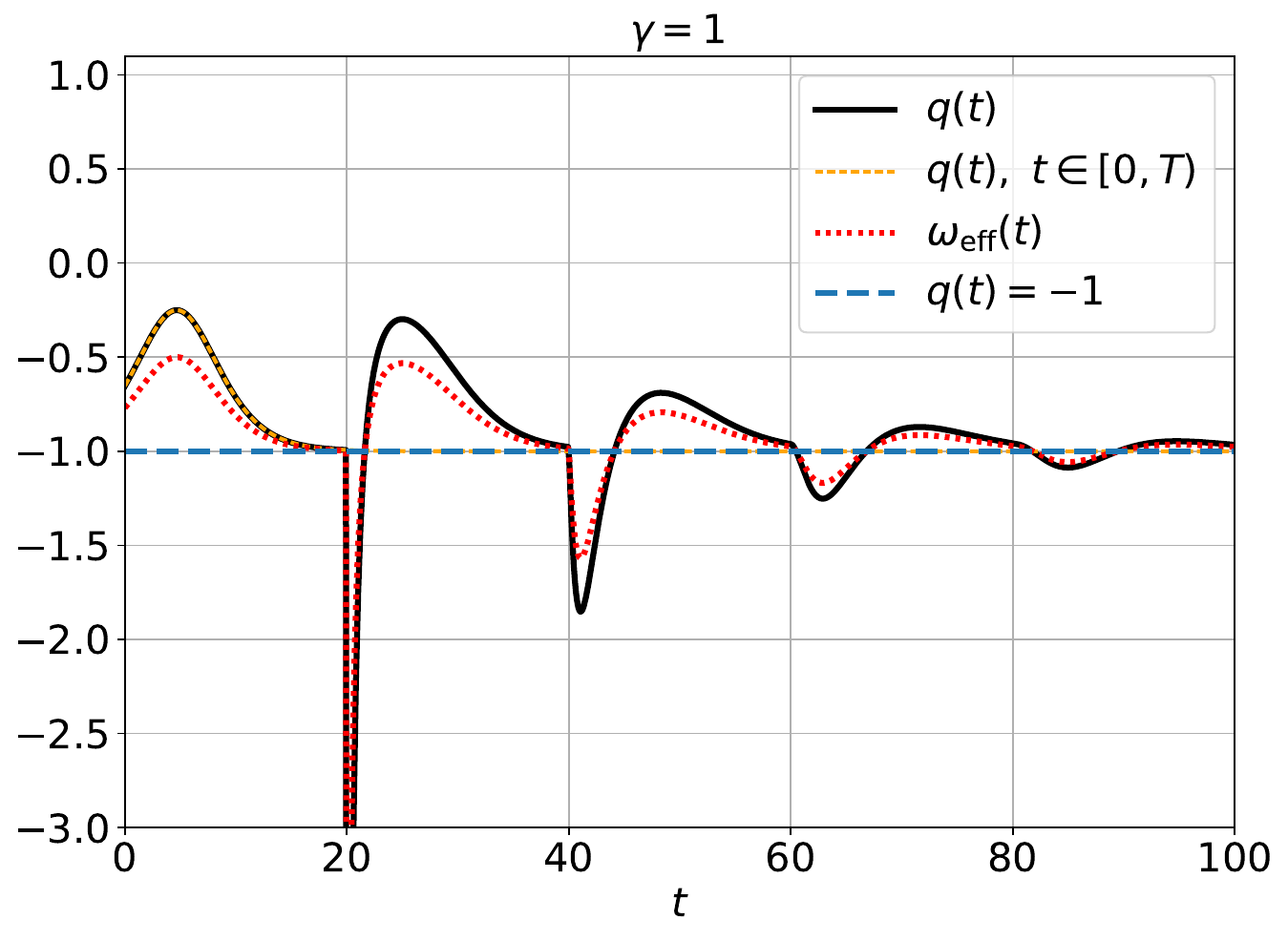}
\caption{\label{Fig0ObsB} Analytical $q(t)$ and $\omega_\text{eff}(t)$ given by \eqref{q_2} and \eqref{weff_2} respectively, for the cases $\gamma=4/3,1$. The other parameters are $\eta_0=0.2$, $T=20$ and $H_0=1$. The minimun values of $q(t)$ and $\omega_\text{eff}(t)$ are $(q=-18.9,\omega_\text{eff}=-12.9)$ for the case $\gamma=4/3$ (radiation), and $(q=-10.7,\omega_\text{eff}=-7.47)$ for the case $\gamma=1$ (matter).}
\end{figure}

\subsection{Discussion}
For the interval $[0,T)$, the definitions \eqref{H-ini}, \eqref{q-ini}, and \eqref{w-ini} describe $H(t)$, $q(t)$, and $w_{\text{eff}}(t)$ as outlined in Proposition \ref{Prop_7b}. 

Based on Propositions \ref{Prop_6} and \ref{Prop_7}, the asymptotic behavior is:
\[
\lim_{t \to +\infty} H(t) = \frac{2 \eta_0}{3 \gamma}, \quad \lim_{t \to +\infty} q(t) = -1, \quad \lim_{t \to +\infty} w_{\text{eff}}(t) = -1,
\]
corresponding to the de Sitter solution.  
This solution emerges after a finite number of phantom epochs, where the effective equation of state satisfies $w_{\text{eff}}(t) < -1$, as illustrated by the numerical results in Figure \ref{Fig0Obs}.
For $\gamma = 4/3$, the deceleration parameter takes on both positive and negative values, while for $\gamma = 1$, it remains consistently negative. The viscosity, modeled by $H$ and evaluated at the retarded time, generates the negative pressure necessary to accelerate the expansion of the universe. Consequently, at late times, both the deceleration parameter and the effective equation of state parameter converge to the expected values characteristic of a de Sitter space-time. These numerical results align with and support Proposition \ref{Prop_7}.
\section{Time-delayed bulk viscosity in Fractional Cosmology}\label{sect.4}
In this section, we promote equation \eqref{nonfractionalandretarded} to the fractional version, which we will solve analytically:
\begin{equation}\label{eqtosolve}
    \leftindex_{}^{\text{C}}D_t^\alpha y(t)=c_1y(t)+c_2y(t-T), \quad y(t)=0 \quad \forall t<0,
\end{equation}
where $\displaystyle \leftindex_{}^{\text{C}}D_t^\alpha$ is the Caputo derivative of order $\alpha$, and in our case, $c_1=-2\eta_0$ and $c_2=\eta_0$. We need to find the solution to this time-delayed fractional differential equation.
\subsection{Problem setting}\label{sect.2}
Our master equation \eqref{eqtosolve} belongs to the class of fractional differential equations:
\begin{equation}
\begin{split}
& {}^C D_t^\alpha y(t) + a y(t - T) + b y(t) = 0, \quad y(t)=0, \; t<0, \\
& y(0) = y_0, \quad y'(0) = y_1, \quad \ldots, \quad y^{(n-1)}(0) = y_{n-1}, \quad n-1<\alpha < n,
\end{split}
\end{equation}
with parameters:
\begin{align*}
\alpha &: \text{order of the fractional derivative}, \\
a &: \text{constant coefficient of the delayed term}, \\
b &: \text{constant coefficient of the linear term}, \\
T &: \text{time delay}.
\end{align*}
\textbf{Derivation Steps}
\begin{enumerate}
\item Start with the differential equation: \({}^C D_t^\alpha y(t) + a y(t - T) + b y(t) = 0.\)
\item Apply the Laplace transform: \(
\mathcal{L}\{{}^C D_t^\alpha y(t)\} + \mathcal{L}\{a y(t - T)\} + \mathcal{L}\{b y(t)\} = 0.
\)
\item Laplace transform of the Caputo fractional derivative: 
\newline\(
\mathcal{L}\{{}^C D_t^\alpha y(t)\} = s^\alpha Y(s) - \sum_{k=0}^{n-1} s^{\alpha-k-1} y^{(k)}(0),\)
where \(Y(s)\) is the Laplace transform of \(y(t)\) and \(y^{(k)}(0)\) are the initial conditions.
\item Laplace transform of the delayed term: \(\mathcal{L}\{y(t - T)\} = e^{-sT} Y(s).\)
\item Laplace transform of the linear response term: \(\mathcal{L}\{y(t)\} = Y(s).\)
\item Substitute into the original equation:
\[
s^\alpha Y(s) - \sum_{k=0}^{n-1} s^{\alpha-k-1} y^{(k)}(0) + a e^{-sT} Y(s) + b Y(s) = 0.
\]
\item Combine terms: \(\left( s^\alpha + a e^{-sT} + b \right) Y(s) - \sum_{k=0}^{n-1} s^{\alpha-k-1} y^{(k)}(0) = 0.\)
\item Characteristic equation: \(\left( s^\alpha + a e^{-sT} + b \right) Y(s) = \sum_{k=0}^{n-1} s^{\alpha-k-1} y^{(k)}(0).\)
\end{enumerate}
The objective is to explore a model using inverse transforms, focusing on its mathematical and physical properties. The investigation involves the following steps:
\begin{enumerate}
\item\textbf{Inverse Transforms}: Inverse transforms, such as the inverse Laplace transform, are powerful tools for solving differential equations. They convert complex differential equations into simpler algebraic forms, making them easier to analyze. Once solutions are obtained in the transformed domain, inverse transforms convert them back to the original domain. This approach helps in understanding the system's behavior.
\item\textbf{Characteristic equation}: The characteristic equation is derived from the differential equation governing the system. It encapsulates the system's key properties and helps determine its stability, oscillatory behavior, and response to external stimuli. By examining the roots of the characteristic equation, we gain insights into the system's dynamics and can predict its long-term behavior.
\item\textbf{Physical Application}: Time-Delayed Bulk Viscosity Cosmology. Consider applying the model to time-delayed bulk viscosity cosmology as a practical example. Bulk viscosity refers to the resistance of cosmic fluids to compression, affecting the Universe's expansion rate. Time delays account for the finite response time of these fluids to changes in pressure and density.
\end{enumerate}
We can develop a more accurate representation of the Universe's evolution by incorporating time delays and bulk viscosity into cosmological models. This approach allows us to:
\begin{enumerate}
\item \textbf{Capture Delayed Reactions}: Time delays introduce memory effects, meaning the system's current state depends on its past states. That is crucial for modeling realistic physical systems where changes do not happen instantaneously.
\item \textbf{Analyze stability}: The characteristic equation provides information about the stability of the cosmological model. By examining the roots, we can determine whether the Universe's expansion will be stable, oscillatory, or exhibit other behaviors.
\item \textbf{Predict Cosmic Evolution}: We can predict how the Universe's expansion rate evolves by solving the time-delayed differential equations. That can help address unresolved issues in cosmology, such as the nature of dark energy and the mechanisms driving accelerated expansion.
\end{enumerate}
Investigating this model using inverse transforms, examining the characteristic equation, and exploring a physical application, such as time-delayed bulk viscosity cosmology, offers a comprehensive approach to understanding complex systems. This method provides valuable insights into the model's mathematical structure and physical behavior, contributing to our knowledge of cosmological dynamics. 

\subsection{Solution}
The Caputo derivative is defined as~\eqref{CaputoD}  where $\alpha\in\mathbb{R}$ and $n\in\mathbb{Z}$. 
Calculating the Laplace transform of the Caputo derivative:  
\begin{equation}
\begin{split}
    \mathcal{L}\left\{\leftindex_{}^{\text{C}}D_t^\alpha y(t)\right\}&=\frac{1}{\Gamma(n-\alpha)}\int_0^\infty \left[\int_0^t\frac{d^n y(\tau)}{d\tau^n}\cdot(t-\tau)^{n-1-\alpha}\,d\tau \right]e^{-st}\,dt.
\end{split}
\end{equation}  
We see that $0\leq t<\infty$ and $0\leq \tau \leq t$, and following the steps of Appendix \ref{app2} we apply the Laplace transform to the equation \eqref{eqtosolve}, and using the Laplace transform of a delayed function \eqref{LaplaceTdelayedf}, we get
\begin{equation}
 s^\alpha \mathcal{L}\left\{y(t)\right\}-\sum _{k=0}^{n-1}s^{\alpha-k-1}y^{(k)}(0)=c_1\mathcal{L}\left\{y(t)\right\}+c_2e^{-sT}\mathcal{L}\left\{y(t)\right\}.
\end{equation}
We can solve for $\mathcal{L}\left\{y(t)\right\}$:
\begin{equation}
\begin{split}
 \mathcal{L}\left\{y(t)\right\}\left(s^\alpha-c_1-c_2e^{-sT}\right)&=\sum_{k=0} ^{n-1}s^{\alpha-k-1}y^{(k)}(0) \implies 
 \mathcal{L}\left\{y(t)\right\} =\frac{\sum _{k=0}^{n-1}s^{\alpha-k-1}y^{(k)} (0)}{s^\alpha-c_1-c_2e^{-sT}}.
\end{split}
\end{equation}
Considering $0<\alpha<1$,  and combining the steps, we get:
\begin{align*}
\mathcal{L}\{y(t)\} &= \frac{s^{\alpha-1} y(0)}{s^\alpha - c_1 - c_2 e^{-sT}} = \frac{y(0)}{s} \left( 1 - \frac{c_1 + c_2 e^{-sT}}{s^\alpha} \right)^{-1} \\
&= \frac{y(0)}{s} \sum_{j=0}^{\infty} \frac{\left(c_1 + c_2 e^{-sT}\right)^j}{s^{\alpha j}}\quad \text{for} \; 0 < \left|\frac{c_1 + c_2 e^{-sT}}{s^\alpha}\right| < 1.
\end{align*}
Using the Newton binomial, we have
\begin{equation}
    \left(c_1+c_2e^{-sT}\right)^j=\sum_{k=0}^j\frac{j!}{k!(j-k)!}c_1^{j-k}c_2^ke^{-skT}.
\end{equation}
The Laplace transform becomes,
\begin{equation}
\begin{split}
    \mathcal{L}\left\{y(t)\right\}&=y(0)\sum_{j=0}^\infty\sum_{k=0}^j\frac{j!}{k!(j-k)!}c_1^{j-k}c_2^k\frac{e^{-skT}}{s^{\alpha j+1}}.
\end{split}
\end{equation}
The inverse Laplace transform gives
\begin{equation}
\begin{split}
    y(t)&=y(0)\sum_{j=0}^\infty\sum_{k=0}^j\frac{j!}{k!(j-k)!}c_1^{j-k}c_2^k\mathcal{L}^{-1}\left[\frac{e^{-skT}}{s^{\alpha j+1}}\right]\\
&=y(0)\sum_{j=0}^\infty\sum_{k=0}^j\frac{j!}{k!(j-k)!}c_1^{j-k}c_2^k\frac{(t-kT)^{\alpha j}\theta(t-kT)}{\Gamma(\alpha j+1)}
\\ 
    &=y(0)\sum_{k=0}^\infty\sum_{j=k}^\infty\frac{j!}{k!(j-k)!}c_1^{j-k}c_2^k\frac{(t-kT)^{\alpha j}\theta(t-kT)}{\Gamma(\alpha j+1)}.
\end{split}
\end{equation}
Remembering that $c_1=-2\eta_0$ and $c_2=\eta_0$, the solution is
\begin{equation}\label{Analyticaly2}
    y(t)=y(0)\sum_{k=0}^\infty\sum_{j=k}^\infty\frac{j!}{k!(j-k)!}(-2)^{j-k}\eta_0^j\frac{(t-kT)^{\alpha j}\theta(t-kT)}{\Gamma(\alpha j+1)},
\end{equation}
which assuming $y(0)=H_0-H_B$ leads to 
\begin{equation}\label{sol}
    H(t)=H_B+\left(H_0-H_B\right)\sum_{k=0}^\infty\sum_{j=k}^\infty\frac{j!}{k!(j-k)!}(-2)^{j-k}\eta_0^j\frac{(t-kT)^{\alpha j}\theta(t-kT)}{\Gamma(\alpha j+1)}.
\end{equation}
We can notice that in the limit  $\alpha\to 1$, we recover the solution without a fractional derivative given by \eqref{AnalyticalH1}.
\begin{Remark}\label{Rem_7}
    For each $t>0$, the external series in \eqref{Analyticaly2} is a finite sum. To see this, note that $\theta(t-kT)=0$ for all $k>t/T$. Then, 
    \[y(t)=\left(H_0-H_B\right)\sum_{k=0}^{\lfloor t/T\rfloor}\sum_{j=k}^\infty\frac{j!}{k!(j-k)!}(-2)^{j-k}\eta_0^j\frac{(t-kT)^{\alpha j}\theta(t-kT)}{\Gamma(\alpha j+1)}.\]
    Furthermore, if we divide the time domain into intervals of length $T$, for each $t$, there exists an $n\in\mathbb{N}_0$, such that $t\in [nT,(n+1)T)$, and  $\lfloor t/T\rfloor=n$.
\end{Remark}
\begin{Proposition}\label{Prop_8}
For each $t>0$,  the solution for 
    \begin{equation}\label{equation(54)}
    \leftindex_{}^{\text{C}}D_t^\alpha y(t)=-2\eta_0 y(t)+\eta_0 y(t-T), \quad y(t)=0 \quad \forall t<0 
\end{equation} 
is  \[y(t)=\left(H_0-H_B\right)\sum_{k=0}^{\lfloor t/T\rfloor}\sum_{j=k}^\infty\frac{j!}{k!(j-k)!}(-2)^{j-k}\eta_0^j\frac{(t-kT)^{\alpha j}\theta(t-kT)}{\Gamma(\alpha j+1)}.\]
For the initial interval $t\in [0, T)$, recalling the definition of the Mittag-Leffler function \eqref{1.1.2.3},
\begin{equation}
y(t)=\left(H_0-H_B\right)\sum_{j=0}^\infty \left(-2 \eta_0 t^{\alpha} \right)^{j} \frac{1}{\Gamma(\alpha j+1)}= \left(H_0-H_B\right) E(\alpha,-2\eta_0 t^{\alpha}).\label{equation(56)}
\end{equation}
\end{Proposition}
\begin{proof} Proposition \ref{Prop_8} is proven directly by applying Remark \ref{Rem_7}. \end{proof}

We analyze the convergence of the partial sums:
\begin{equation}
S_n(t) = \sum_{k=0}^{\lfloor t/T\rfloor}\sum_{j=k}^n \frac{j!}{k!(j-k)!} (-2)^{j-k} \eta_0^j \frac{(t-kT)^{\alpha j} \theta(t-kT)}{\Gamma(\alpha j+1)} \quad \text{as} \quad n \to \infty.
\end{equation}
\begin{Proposition}\label{Prop_9}
  For each $t>0$, 
\begin{align}\label{solC}
    H(t)& =H_B+\left(H_0-H_B\right) \lim_{n \rightarrow \infty}\sum_{k=0}^{\lfloor t/T\rfloor}\sum_{j=k}^n\frac{j!}{k!(j-k)!}(-2)^{j-k}\eta_0^j\frac{(t-kT)^{\alpha j}\theta(t-kT)}{\Gamma(\alpha j+1)},
\\
   a(t) & =  e^{ H_B t} \lim_{n \rightarrow \infty} \prod_{k=0}^{\lfloor t/T\rfloor} \prod_{j=k}^n  \exp \left(\frac{ \eta_{0}^j (-2)^{j-k} \Gamma (j+1) \left(H_0-H_B\right) \theta (t-k T) (t-k T)^{\alpha  j+1}}{\Gamma (k+1) \Gamma (j \alpha +2)
   \Gamma (j-k+1)}\right).
\end{align}
\end{Proposition} 
Moreover, from \eqref{weffa} and \eqref{weff}, we have 
\begin{Proposition}
For each $t>0$, 
 
    \begin{align}
    & q =  -1-\frac{\left(H_0-H_B\right) \lim_{n \rightarrow \infty}\sum_{k=0}^{\lfloor t/T\rfloor}\sum_{j=k}^n\frac{j!}{k!(j-k)!}(-2)^{j-k}\eta_0^j \alpha j \frac{(t-kT)^{\alpha j-1}\theta(t-kT)}{\Gamma(\alpha j+1)}}{{\left(H_B+\left(H_0-H_B\right) \lim_{n \rightarrow \infty}\sum_{k=0}^{\lfloor t/T\rfloor}\sum_{j=k}^n\frac{j!}{k!(j-k)!}(-2)^{j-k}\eta_0^j\frac{(t-kT)^{\alpha j}\theta(t-kT)}{\Gamma(\alpha j+1)}\right)}^2}, \label{q-model}
    \\
    & w_{\text{eff}}  =(2 q -1)/3. \label{w-model}
\end{align}

For the initial interval $t\in (0, T)$, 
 
  \begin{align}
    q =-1-\frac{\left(H_0-H_B\right) \frac{d}{d t}E(\alpha, -2\eta_0 t^{\alpha} )}{{\left(H_B+\left(H_0-H_B\right) E(\alpha, -2\eta_0 t^{\alpha} )\right)}^2}  & = -1+\frac{6 \gamma  \eta_0 t^{\alpha-1} (3 \gamma  H_0-2 \eta_0) E\left(\alpha ,\alpha, -2 t^{\alpha } \eta_0\right)}{\left(2 \eta_0+(3\gamma  H_0-2 \eta_0)E\left(\alpha, -2 t^{\alpha }
   \eta_0\right)\right){}^2},\label{ini_q-model}\\ 
    w_{\text{eff}} & = -1+\frac{4 \gamma  \eta_0 t^{\alpha-1} (3 \gamma  H_0-2 \eta_0) E\left(\alpha ,\alpha, -2 t^{\alpha } \eta_0\right)}{\left(2 \eta_0+(3\gamma  H_0-2 \eta_0)E\left(\alpha, -2 t^{\alpha }
   \eta_0\right)\right){}^2}. \label{ini_w-model}
\end{align}

\end{Proposition}

\subsection{Error and Smooth Transition Correction}
Now, we specify the number of terms used in the series expansion and discuss the criteria for convergence. Additionally, we provide details on computational runtime, round-off errors, and measures taken to ensure numerical stability.

We define the function:
\begin{equation}\label{solB}
H_n(t) = H_B + \left(H_0 - H_B\right) \sum_{k=0}^{\lfloor t/T\rfloor}\sum_{j=k}^n \frac{j!}{k!(j-k)!} (-2)^{j-k} \eta_0^j \frac{(t-kT)^{\alpha j} \theta(t-kT)}{\Gamma(\alpha j+1)}.
\end{equation}
The round-off error in approximating \( H(t) \) by truncating the series at \( n \) is given by:
\begin{align}
H(t)-H_n(t) &= \left(H_0 - H_B\right) \sum_{k=0}^{\lfloor t/T\rfloor} \frac{\theta(t-kT)}{(-2)^k k!} \sum_{j=n+1}^{\infty} \frac{j!}{(j-k)!} (-2\eta_0)^{j} \frac{(t-kT)^{\alpha j}}{\Gamma(\alpha j+1)}.
\end{align}
Expanding the falling factorial,
\begin{equation}
\frac{j!}{(j-k)!} = j \cdot (j-1) \cdots (j-k+1),
\end{equation}
which asymptotically simplifies to:
\begin{equation}
\frac{j!}{(j-k)!} \sim j^k, \quad \text{for large } j.
\end{equation}
Approximating the gamma function using Stirling's formula, we apply upper and lower bounds given in \cite{Karatsuba2001, Mortici2011a, Mortici2011b, Mortici2011c}:
\begin{equation}
\sqrt{\pi} \left(\frac{x}{e}\right)^x \left( 8x^3 + 4x^2 + x + \frac{1}{100} \right)^{1/6} < \Gamma(1+x) <  \sqrt{\pi} \left(\frac{x}{e}\right)^x \left( 8x^3 + 4x^2 + x + \frac{1}{30} \right)^{1/6}.
\end{equation}
Taking the reciprocal:
\begin{equation}
\frac{1}{\Gamma(1+x)} \sim\left(\frac{e}{x}\right)^x \left(\frac{\sqrt{\frac{1}{x}}}{\sqrt{2 \pi}} + \mathcal{O}\left(\left(\frac{1}{x}\right)^{3/2}\right)\right).
\end{equation}
By substituting \( x = \alpha j \), we approximate:
\begin{equation}
\frac{j!}{(j-k)! \cdot \Gamma(\alpha j + 1)} \sim j^k \cdot \left(\frac{e}{{\alpha j}}\right)^{\alpha j} \frac{1}{\sqrt{2 \pi \alpha j}}.
\end{equation}
Expanding the denominator,
\begin{equation}
\left(\frac{e}{\alpha j}\right)^{\alpha j} = (\alpha j)^{-\alpha j} \cdot e^{\alpha j},
\end{equation}
which simplifies further:
\begin{equation}
\frac{j!}{(j-k)! \cdot \Gamma(\alpha j + 1)} \sim \frac{j^k \cdot e^{\alpha j}}{\sqrt{2\pi \alpha j} \cdot (\alpha j)^{\alpha j}}.
\end{equation}
Thus, the final asymptotic approximation:
\begin{equation}
\frac{j!}{(j-k)! \cdot \Gamma(\alpha j + 1)} \sim \frac{e^{\alpha j} \cdot j^{k - \alpha j}}{\sqrt{2\pi \alpha j} \cdot \alpha^{\alpha j}}.
\end{equation}
Hence,
\begin{equation}
\frac{H(t)-H_n(t)}{H_0 - H_B} \sim  \sum_{k=0}^{\lfloor t/T\rfloor} \frac{\theta(t-kT)}{(-2)^k k!} \sum_{j=n+1}^{\infty} \frac{e^{\alpha j} j^{k - \alpha j}}{\sqrt{2\pi \alpha j} \alpha^{\alpha j}} (-2\eta_0)^{j} (t-kT)^{\alpha j}, \quad \text{for large } n.
\end{equation}
Breaking down the terms:
\begin{itemize}
    \item \( \lfloor t/T \rfloor \): Integer division of \( t \) by \( T \), which determines the upper bound of summation.
    \item \( \theta(t - kT) \): Heaviside step function, restricting contributions to cases where \( t \geq kT \).
    \item \( \frac{1}{k! (2)^k} \): A rapidly decaying term as \( k \) increases, primarily driven by factorial growth in the denominator.
    \item \( e^{\alpha j} \): An exponentially growing term, significant for large \( j \).
    \item \( j^{k-\alpha j} \): A power-law term dependent on \( k - \alpha j \), which decays when \( \alpha j > k \).
    \item \( (-2\eta_0)^j \): An exponential factor determined by \( \eta_0 \), growing or decaying based on \(|2\eta_0|\).
    \item \( (t - kT)^{\alpha j} \): A power-law term influenced by \( (t - kT) \):
    \begin{itemize}
        \item If \( t - kT > 1 \), this term increases with \( \alpha j \).
        \item If \( 0 < t - kT < 1 \), it decreases with \( \alpha j \).
        \item If \( t - kT = 1 \), it remains constant.
    \end{itemize}
    \item \( \sqrt{2\pi \alpha j} \): A polynomial decay factor in the denominator, which is subleading.
    \item \( \alpha^{\alpha j} \): A dominant exponential decay term in the denominator for large \( j \).
\end{itemize}
By combining all terms, the dominant contribution for large \( n \) comes from \( k = 0 \):
\begin{equation}
H(t) - H_n(t) \sim \left(H_0 - H_B\right) (-1)^{n+1} \frac{1}{\sqrt{2\pi \alpha (n+1)}} \cdot \left[\frac{e (2\eta_0)^{\frac{1}{\alpha}} \cdot t }{\alpha(n+1)} \right]^{\alpha (n+1)}. \label{error}
\end{equation}
The error approaches zero as \( n \) increases, provided that:
\begin{equation}
0 < t < t(n) = \frac{\alpha(n+1)}{e (2\eta_0)^{\frac{1}{\alpha}} }.
\end{equation}
Under these conditions, we obtain:
\begin{equation}
H(t) = \lim_{n\rightarrow \infty} H_n(t), \quad 0 < t < \infty.
\end{equation}
For a given $n$, the time to achieve the tolerance is 
\[t_{\text{Tol}}(n)=  t(n) \left(\frac{\sqrt{2 \pi \alpha (n+1)}\text{Tolerance}}{ |H_0 - H_B|}\right)^{\frac{1}{\alpha (n+1)}}\]
Defining $t_{\text{Final}} = \max\{t(n), t_{\text{Tol}}(n)\}$, a reasonable stopping condition is 
\[t_{\text{Final}}< t_{\text{Limit}}\quad \text{and} \quad |H_0 - H_B| \cdot \frac{1}{\sqrt{2 \pi \alpha (n+1)}} \cdot \left(\frac{t}{t(n)}\right)^{\alpha (n+1)} > \text{Tolerance}.\]
Therefore, $|H(t) - H_n(t)|< \text{Tolerance}$ for $t>t_{\text{Final}}$ and $n>n_{\text{max}}$. 

To ensure a smooth solution, we introduce the mollifier:
\begin{equation}
S(t, t_{\text{Final}}) = 1 - \exp\left(-\left(\frac{t - t_{\text{Final}}}{T}\right)^2\right).\label{mollifier}
\end{equation}
Mollifiers are functions used in numerical analysis and differential equations to ensure smooth transitions between values, preventing discontinuities \cite{lions1796quelques,brezis2011functional}. They allow:   
\begin{itemize}
    \item \textbf{Smooth Transition:} The exponential decay gradually introduces the correction as \( t \) exceeds \( t_{\text{Final}} \).
    \item \textbf{Bounded Between 0 and 1:} The function transitions smoothly from 0 at \( t = t_{\text{Final}} \) and asymptotically approaches 1 as \( t \) increases.
    \item \textbf{Gaussian-like Decay:} The squared term in the exponent resembles a Gaussian mollifier, controlling the transition rate.
    \item \textbf{Discontinuity Prevention:} This function ensures corrections are applied smoothly rather than abruptly, preserving continuity.
\end{itemize}
Thus, the final computational solution is:
\begin{equation}
H_n^{\text{corrected}}(t)= H_n(t) + S(t, t_{\text{Final}}) \left(H_0 - H_B\right) (-1)^{n+1} \frac{1}{\sqrt{2\pi \alpha (n+1)}} \cdot \left[\frac{t }{t(n)} \right]^{\alpha (n+1)}. \label{H_summation}
\end{equation}
Appendix \ref{sect_mollifier} presents an algorithm for approximating \( H(t) \) using equation \eqref{H_summation} with the mollifier \eqref{mollifier}. Figure \ref{Fig0} illustrates analytical \( H_m(t) \) for \( m=500 \) from \eqref{solB}, with parameters \( \alpha=0.9 \), \( \gamma=4/3,1 \), \( \eta_0=0.2 \), \( T=20 \), and \( H_0=1 \). 
Using \( H_n^{\text{corrected}}(t) \) from \eqref{H_summation} for \( n=114 \) ensures best accuracy, satisfying 
\(
|H(t)-H_n^{\text{corrected}}(t)| <
\begin{cases} 
0.004, & \text{if } \gamma = \frac{4}{3} \\
0.005, & \text{if } \gamma = 1
\end{cases}
\) with fewer terms
for \( t>t_{\text{max}} \approx 105.39 \), while maintaining \( |H_n^{\text{corrected}}(t)-H_{m}(t)|<0.025 \), significantly reducing computational cost. Figure \ref{Fig00} compares \( H_n^{\text{corrected}}(t) \), defined by \eqref{H_summation} for \( n=114 \), with Mittag-Leffler \( H_B + ( H_0 - H_B ) E (\alpha , - 2 \eta_0 t \alpha ) \), using the same parameters.
\begin{figure}[h]
    \centering
    \includegraphics[width=0.49\linewidth]{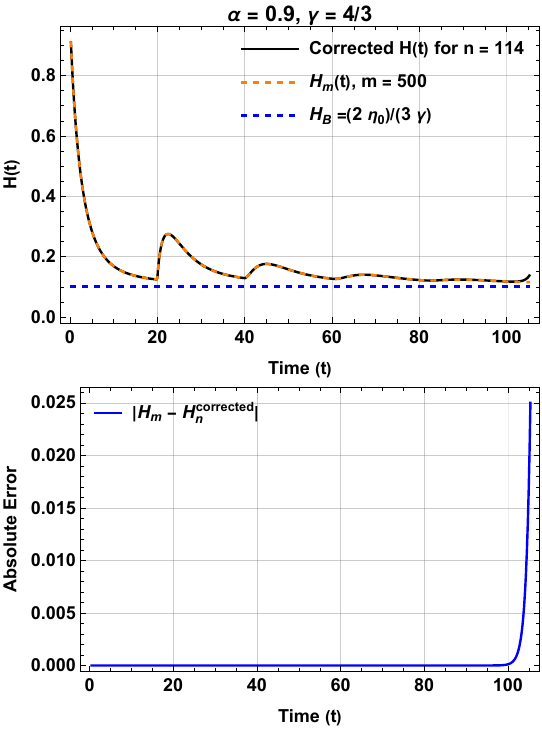}
      \includegraphics[width=0.49\linewidth]{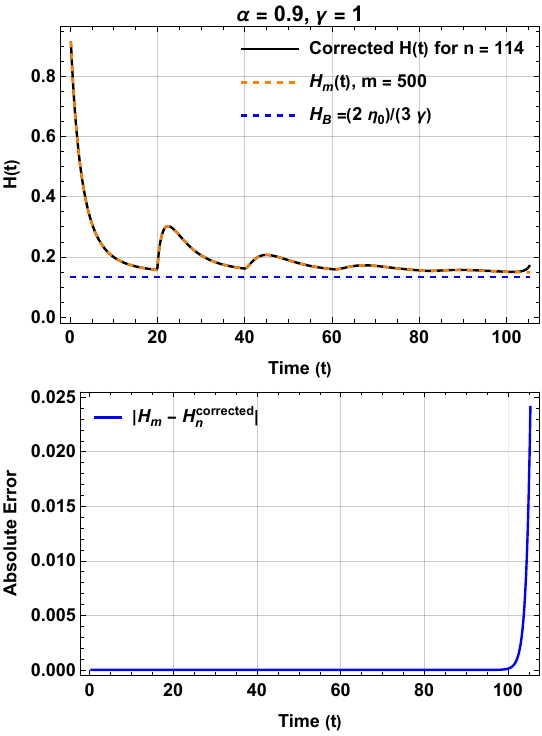}
    \caption{\label{Fig0} Analytical $H_m(t)$ for $m=500$ from \eqref{solB}, with $\alpha=0.9$, $\gamma=4/3,1$ and $\eta_0=0.2$, $T=20$, and $H_0=1$.  Using $H_n^{\text{corrected}}(t)$ defined by \eqref{H_summation} for $n=114$ achieves high accuracy with fewer terms.}
\end{figure}
\begin{figure}[h]
    \centering
    \includegraphics[width=0.49\linewidth]{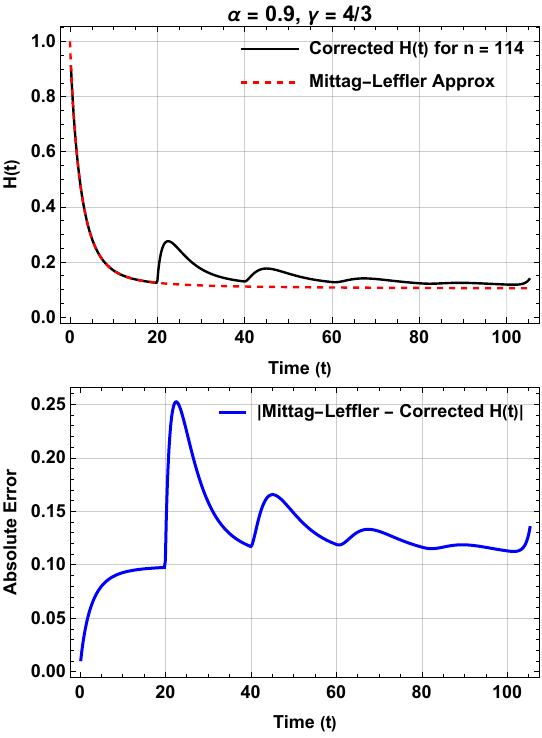}
      \includegraphics[width=0.49\linewidth]{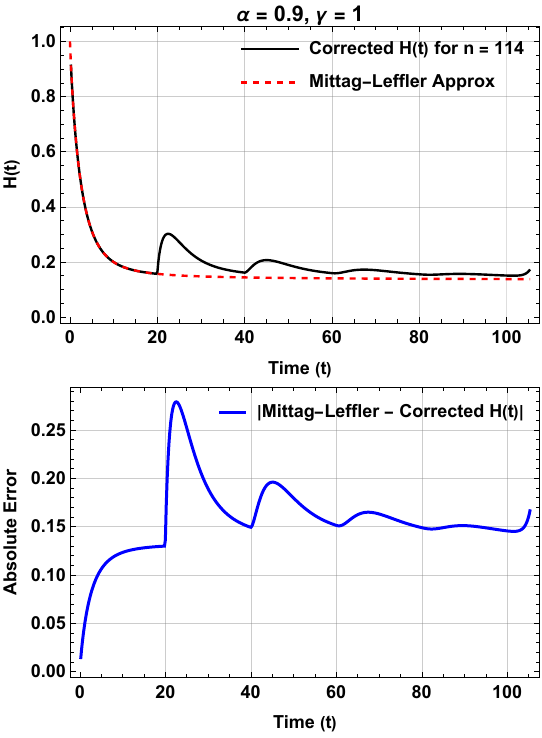}
    \caption{\label{Fig00} $H_n^{\text{corrected}}(t)$ defined by \eqref{H_summation} for $n=114$, with $\alpha=0.9$, $\gamma=4/3,1$ and $\eta_0=0.2$, $T=20$, and $H_0=1$ compared with Mittag-Leffler $H_B +( H_0 - H_B ) E (\alpha , - 2 \eta_0 t \alpha )$.}
\end{figure}
\subsection{Numerical Solution}
\label{sec:numerics}
We want to solve the equation \eqref{eqtosolve} numerically. We use the fractional Euler method for fractional differential equations. The discrete fractional Caputo derivative is given by \cite{ncaputo1}
\begin{equation}\label{ncaputo0}
    ^{\text{C}}D_t^\alpha=\delta^\alpha y_n+R_n,\quad y_n:=y(t_n),\quad t_n=nh, \quad h= T/m, \quad h>0,
\end{equation}
where $m$ is the number of sub-intervals on which the intervals $[k T, (k+1)T), \quad k \in \{0, 1, 2, \ldots\}$ are divided, 
such that integer multiples of $T$, $t_{k\cdot m}= k T$ are on the mesh.

Define
\begin{equation}\label{ncaputo1}
    \delta^\alpha y_n=\frac{h^{-\alpha}}{\Gamma(2-\alpha)}\sum_{i=0}^{n-1}\left[(n-i)^{1-\alpha}-(n-i-1)^{1-\alpha}\right]\left(y_{i+1}-y_i\right),
\end{equation}
and
\begin{equation}
    R_n\sim -\frac{h^{2-\alpha}}{\Gamma(2-\alpha)}\zeta(\alpha-1)y^{\prime\prime}(\tau), \quad \tau\in (0,t),
\end{equation}
where $\zeta$ is the Riemann-Zeta function. 

On the other hand, from \cite{ahmed2018fractional}, we have the formula for the fractional Euler method:
\begin{equation}\label{fract-Euler-Method}
\begin{split}
    y(t_{i+1})& =y(t_i)+\frac{h^\alpha}{\Gamma(\alpha+1)}\left[c_1y(t_i)+c_2y(t_i-T)\right],\\
       y(t_{i+1})& =y(t_i)+\frac{h^\alpha}{\Gamma(\alpha+1)}\left[c_1y(t_i)+c_2y \left(\left(i-m\right)h\right)\right], \\
       y(t_{i+1})& =y(t_i)+\frac{h^\alpha}{\Gamma(\alpha+1)}\left[c_1y(t_i)+c_2 y \left(t_{i-m}\right)\right]. 
    \end{split}
\end{equation}
The calculation of $q$ and  $w_{\text{eff}}$ using the series \eqref{q-model} and \eqref{w-model} is affected by error propagation. Hence, we use a discretized derivative to approximate \( \dot{H}(t) \) by using the forward difference formula (see Appendix \ref{appFDF}):
\begin{equation}
\dot{H}(t_n) \approx \frac{H(t_{n+1}) - H(t_{n})}{h}, \label{Forward}
\end{equation}
where \( t_n \) is the current time step.
But 
\begin{equation}
     H(t_n)=H_B+y(t_n) \implies 
\dot{H}(t_n) \approx \frac{y(t_{n+1}) - y(t_{n})}{h}. \label{euler_Hn}
\end{equation}
where $H_B=H_B$. 
Hence, we calculate \( q \) at \( t_n \), \(q_n \) defined as:
        \begin{equation}
 q_{n} = -1 - \frac{\left(y_{n+1} - y_{n}\right)}{h\left(H_B+y_{n}\right)^2}. \label{euler_qn}
        \end{equation}
 Calculate \( w_{\text{eff}} \) at \( t_n \), \({w_{\text{eff}}}_n \) defined as:
  \begin{equation}
       {w_{\text{eff}}}_n =(2 q_n -1)/3. \label{euler_wn}
        \end{equation}
Investigating the causes of anomalous initial behavior is essential, as it may stem from numerical artifacts or initial conditions influenced by viscosity and time delays. Sensitivity analysis assesses robustness under small perturbations. We applied the forward difference formula \eqref{Forward} for time derivative evaluation, with additional reliable schemes listed in Appendix \ref{appFDF}. 

For implementing this numerical procedure, we are required the initial terms $y_0=y(t_0), y_1=y(t_1), \ldots y_m=y(t_m)$, with $t_0=0, t_1= h, \ldots t_k = k h, \ldots t_m=T$. 
For $t\in [0, T)$, using \eqref{equation(56)}, we have 
\begin{equation}
y(t)= \left(H_0-H_B\right) E(\alpha,-2\eta_0 t^{\alpha}), \quad t\in [0, T). \label{eq81}
\end{equation}
By continuity, $y_0=y(0)=1$ and $y_m=y(T) =  \left(H_0-H_B\right) E(\alpha,-2\eta_0 T^{\alpha})$.

For calculating $q_1=q(t_1), \ldots q_m=q(t_m)$, with $t_1= h, \ldots t_k = k h, \ldots t_m=T$ which belongs to  $t\in (0, T)$, we use \eqref{ini_q-model}, to have 
  \begin{equation}
  q(t_k)=      -1+\frac{6 \gamma  \eta_0 t_k^{\alpha-1} (3 \gamma  H_0-2 \eta_0) E\left(\alpha ,\alpha, -2 t_k^{\alpha } \eta_0\right)}{\left(2 \eta_0+(3\gamma  H_0-2 \eta_0)E\left(\alpha, -2 t_k^{\alpha } \label{eq82}
   \eta_0\right)\right){}^2}.
         \end{equation}
Hence, using \eqref{fract-Euler-Method}, \eqref{euler_Hn}, \eqref{euler_qn}, \eqref{euler_wn}, \eqref{eq81} and \eqref{eq82} we have
\begin{subequations}\label{algorithm-1}
\begin{align}
  & y_0=H_0-H_B, \ldots, y_k=  y_0 E(\alpha,-2\eta_0 {\left(k h\right)}^{\alpha}), \ldots, y_m=  y_0 E(\alpha,-2\eta_0 {T}^{\alpha}),\\
   & y_{n+1} =y_n+\frac{h^\alpha}{\Gamma(\alpha+1)}\left[c_1y_n+c_2 y _{n-m}\right]\\
     & H_{n}=H_B+y_{n},\label{H_n}\\
     &  q_{0} = -1 - \frac{y_0\left(E(\alpha,-2\eta_0  h^{\alpha}) - 1\right)}{h H_0^2},\\
     & q_k=  -1+\frac{6 \gamma  \eta_0 {\left(k h\right)}^{\alpha-1} (3 \gamma  H_0-2 \eta_0) E\left(\alpha ,\alpha, -2 {\left(k h\right)}^{\alpha } \eta_0\right)}{\left[2 \eta_0+(3\gamma  H_0-2 \eta_0)E\left(\alpha, -2 {\left(k h\right)}^{\alpha }
   \eta_0\right)\right]{}^2}, k= 1 \ldots m,\\ 
           & q_n = -1- \frac{h^{\alpha-1}\left(c_1y_n+c_2 y _{n-m}\right)}{\Gamma(\alpha+1)\left[H_B+y_n\right]^2}, \quad n\geq m+1, \label{q_n}\\
     &  {w_{\text{eff}}}_n =(2 q_n -1)/3. \label{w_n}
\end{align}
\end{subequations}
The fundamental algorithm for executing the numerical procedure \eqref{algorithm-1} is detailed in Appendix \ref{AlgorithmB}. 
By solving equation \eqref{eqtosolve} for $y(t)$, with $c_1=-2\eta_0$ and $c_2=\eta_0$; and implementing the numerical procedure we obtain figure \ref{Fig1}, which shows $H(t)$  for the numerical solution using the general formula for the fractional Euler method \eqref{fract-Euler-Method}, which is the linearized version. This figure shows that in a universe dominated by radiation ($\gamma=4/3$) or dust ($\gamma=1$), the system reaches the de Sitter phase after some perturbations due to the memory effects introduced by the retarded time. Figure \ref{Fig0Obs} presents the numerical solutions for \( q(t) \) and \( \omega_{\text{eff}}(t) \) for \( \gamma = \frac{4}{3} \) and \( \gamma = 1 \). 
\begin{figure}[h]
\centering
\includegraphics[width=0.5\textwidth]{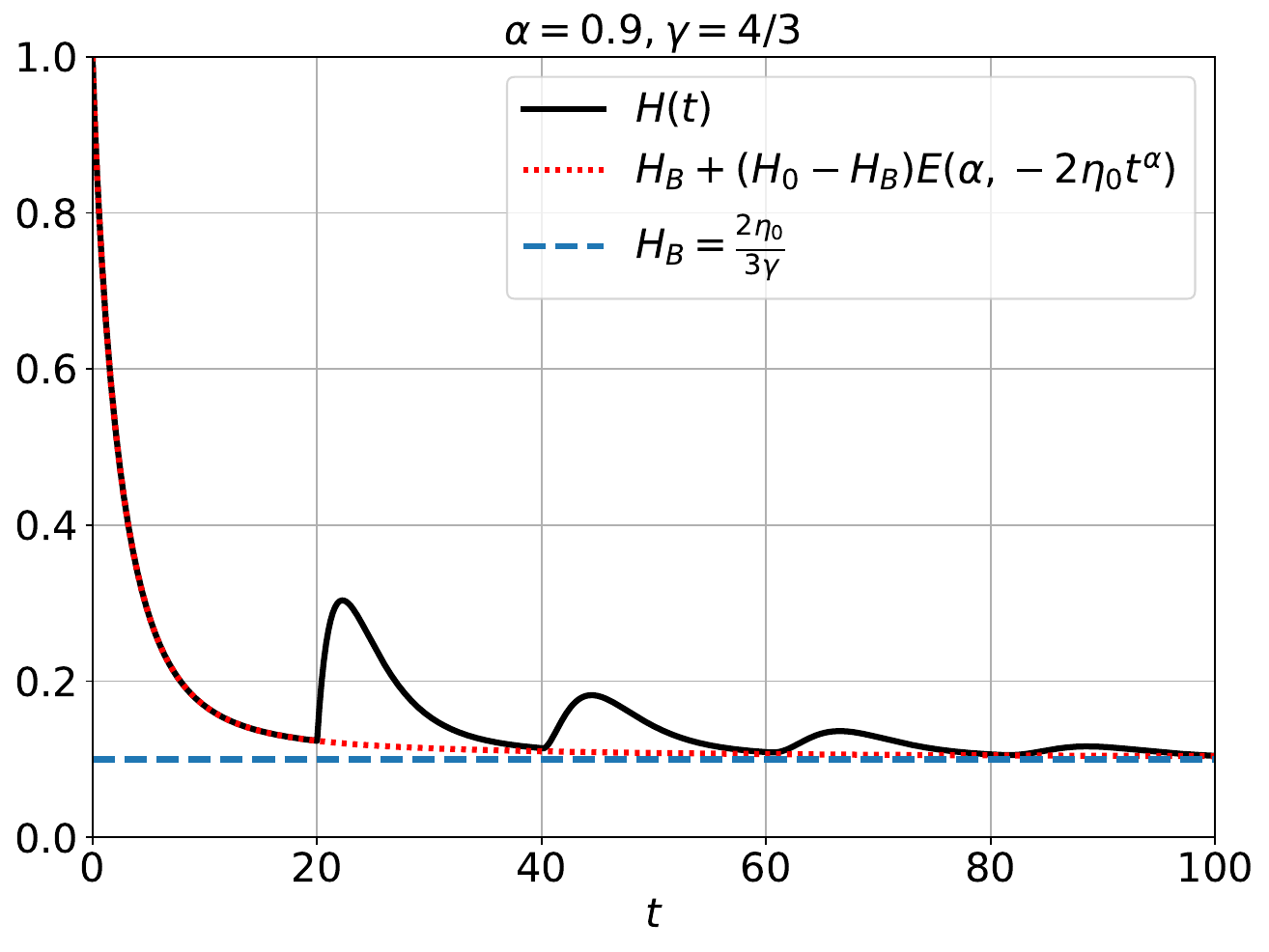}\includegraphics[width=0.5\textwidth]{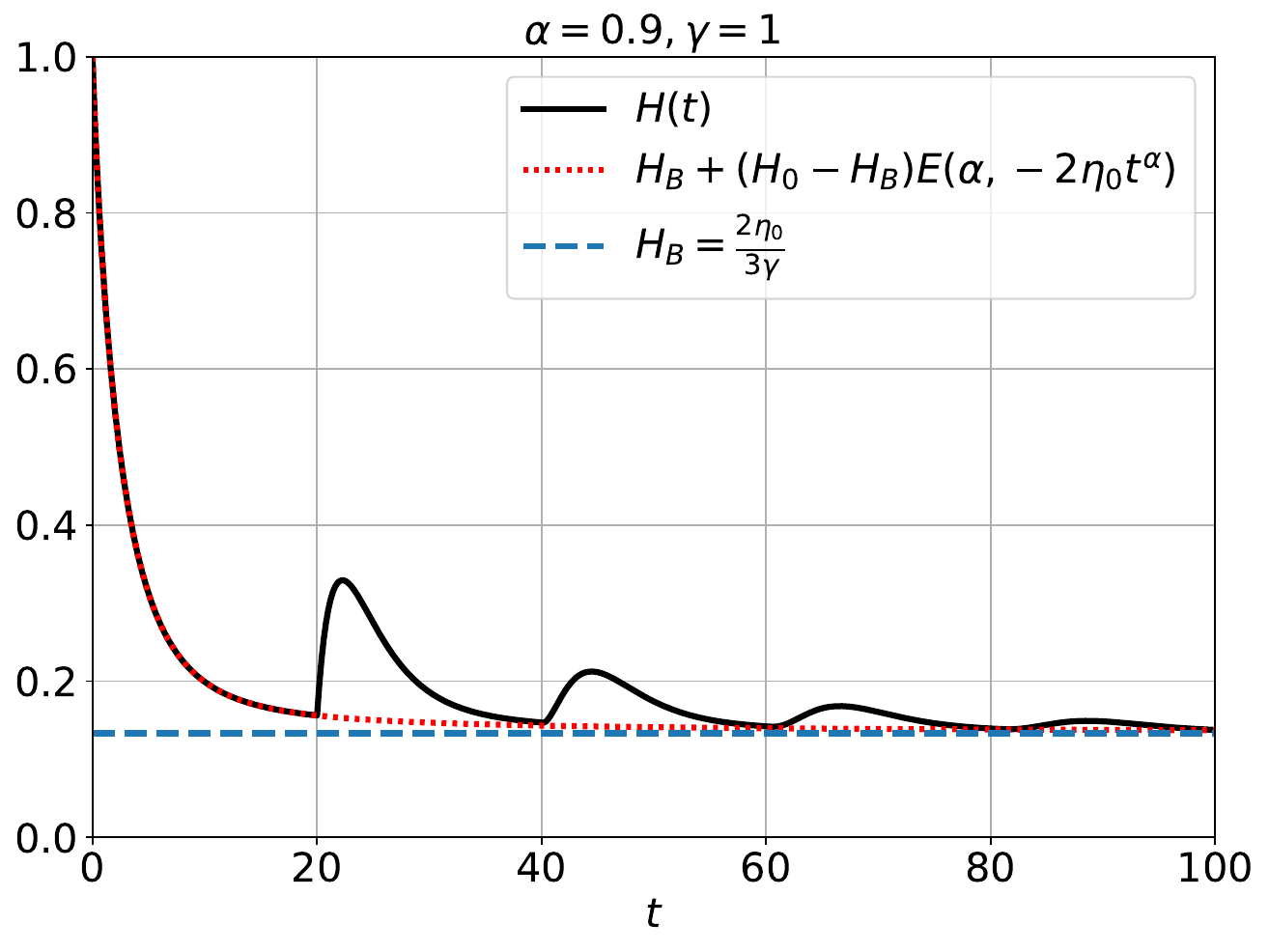}
\caption{\label{Fig1} $H(t)$  for the numerical solution using the formula for the fractional Euler method \eqref{fract-Euler-Method}, for $\alpha=0.9$ and $\gamma=4/3,1$, and Mittag-Leffler function $H_B+E\left(\alpha,-2\eta_0 t^\alpha\right)$. The other parameters are $\eta_0=0.2$, $T=20$ and $H_0=1$. The dashed line represents the de Sitter solution.}
\end{figure}
\begin{figure}[h]
\centering
\includegraphics[width=0.5\textwidth]{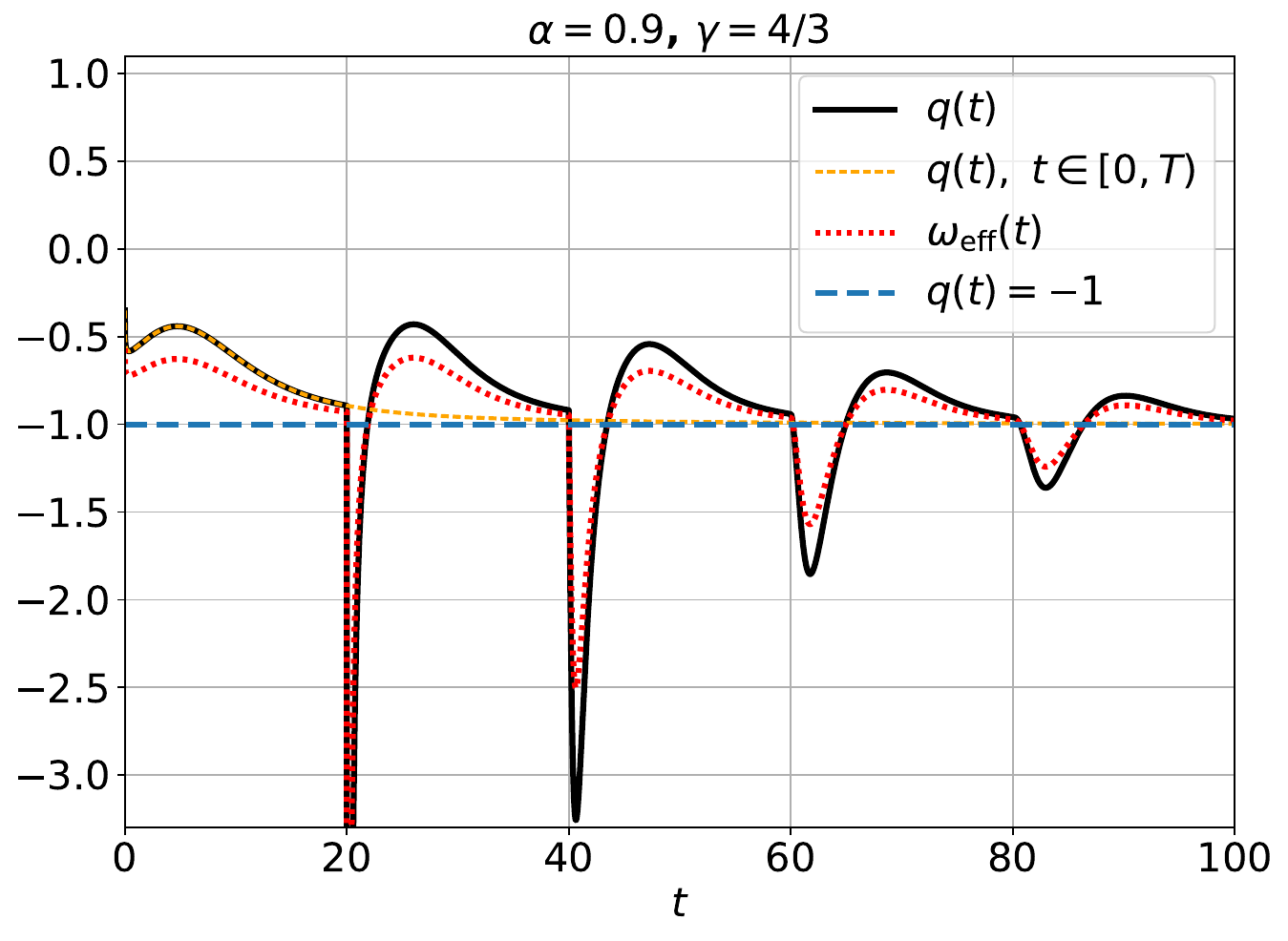}\includegraphics[width=0.5\textwidth]{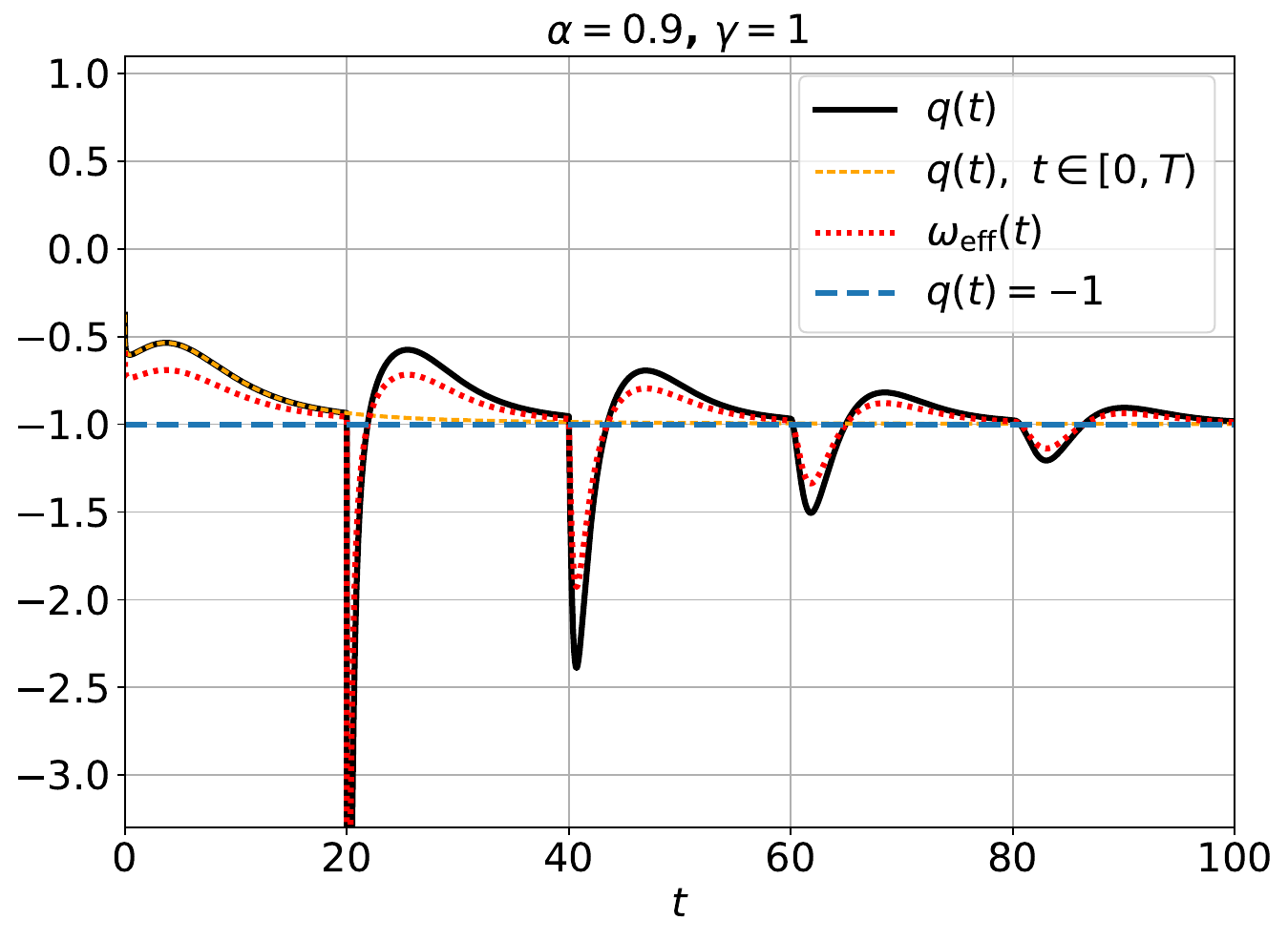}
\caption{\label{Fig0Obs} Numerical solution of the functions $q(t)$ and  $\omega_\text{eff}(t)$, for the cases $\gamma=4/3,1$, considering $\alpha=0.9$. The other parameters are $\eta_0=0.2$, $T=20$ and $H_0=1$. The minimum values of $q(t)$ and $\omega_\text{eff}(t)$ are $(q=-20.6,\omega_\text{eff}=-14.1)$ for the case $\gamma=4/3$ (radiation), and $(q=-12.9,\omega_\text{eff}=-8.95)$ for the case $\gamma=1$ (matter).}
\end{figure}

We can also apply a numerical method to solve the nonlinear equation \eqref{nonfractandnonlinear} in the fractional version:
\begin{equation}\label{fractandnonlinear}
    \leftindex_{}^{\text{C}}D_t^\alpha y(t)=-\frac{3\gamma}{2}y^2(t)-2\eta_0y(t)+\eta_0y(t-T).
\end{equation}
The numerical scheme to solve \eqref{fractandnonlinear} is the following. 
We chose a mesh $y_n:=y(t_n),\quad t_n=nh, \quad h= T/m, \quad h>0$, 
where $m$ is the number of sub-intervals on which the intervals $[k T, (k+1)T), \quad k \in \{0, 1, 2, \ldots\}$ is divided, 
such that integer multiples of $T$, $t_{k\cdot m}= k T$, are on the mesh.

In the interval $[0,T)$ the dynamics is given by the fractional differential equation 
\begin{equation}\label{fractandnonlinearnotDelay}
    \leftindex_{}^{\text{C}}D_t^\alpha y(t)=-\frac{3\gamma}{2}y^2(t)-2\eta_0y(t), \quad t\in[0,T).
\end{equation}
For implementing this numerical procedure, we are required the initial terms $y_0=y(t_0), y_1=y(t_1), \ldots y_m=y(t_m)$, with $t_0=0, t_1= h, \ldots t_k = k h, \ldots t_m=T$:
\begin{subequations}\label{algorithm-2}
\begin{equation}
\begin{split}
 y_{n+1}& =y_n+\frac{h^\alpha}{\Gamma(\alpha+1)}\left[-\frac{3\gamma}{2}y_n^2-2\eta_0 y_n\right], \quad
 H_{n}  =H_B+y_{n}, \\
 q_{n} & = -1 - \frac{\left(y_{n+1} - y_{n}\right)}{h H_n^2}, \quad
   {w_{\text{eff}}}_n = \frac{2q_n - 1}{3}.
 \label{num_notdelay}
\end{split}
 \end{equation}
Results from \eqref{num_notdelay} are used to initialize the delayed procedure given by equation 
\begin{equation}
\begin{split}
 y_{n+1}&=y_n+\frac{h^\alpha}{\Gamma(\alpha+1)}\left[-\frac{3\gamma}{2}y_n^2-2\eta_0 y_n+ \eta_0 y_{n-m}\right], \quad H_n = H_B + y_n, \\ 
    q_n &= -1 + \frac{h^{\alpha-1} \Big(\frac{3\gamma}{2}y_n^2+2\eta_0 y_n-\eta_0 y_{n-m}\Big)}{\Gamma(\alpha+1) H_n^2}, \quad
    {w_{\text{eff}}}_n = \frac{2q_n - 1}{3}.  \label{num_delay}
    \end{split}
 \end{equation}
\end{subequations}
The fundamental algorithm for executing the numerical procedure \eqref{algorithm-2} is detailed in Appendix \ref{AlgorithmC}.

Figure \ref{Fig2} compares the numerical solutions of equations \eqref{nonfractionalandretarded}, \eqref{eqtosolve}, and \eqref{fractandnonlinear} for different $\alpha$ values with $\gamma=4/3$ (radiation). 
The nonlinear solution (dashed red curve) reaches the de Sitter phase faster than the linear one (solid black curve), showing small perturbations over time.
Asymptotically, all solutions approach $H_B+(H_0-H_B)E\left(\alpha,-2\eta_0 t^\alpha\right)$. Increasing $\alpha$ brings the curves closer to the de Sitter state.

In Figure \ref{Fig2.1}, we show the same as Figure \ref{Fig2} but with the value $\gamma=1$ (matter). The behavior is similar to the $\gamma=4/3$ (radiation) case.

In Figure \ref{Fig0Obs0}, we present the numerical solutions for the functions $q(t)$ and $\omega_\text{eff}(t)$ in the non-linear case \eqref{fractandnonlinearnotDelay}, for the values $\gamma=4/3$ and $\gamma=1$, considering $\alpha=0.9$. The other parameters are $\eta_0=0.2$, $T=20$, and $H_0=1$. 

Summarizing, the functions $q(t)$ and $\omega_\text{eff}(t)$ oscillate between positive and negative values and, at late times, converge to the values corresponding to de Sitter spacetime. The effect of the Caputo derivative is that the deceleration parameter and \(\omega_{\text{eff}}\) converge to negative values and the universe remains accelerated for $t\geq T$.
\begin{figure}[h]
\centering
\includegraphics[width=1.0\textwidth]{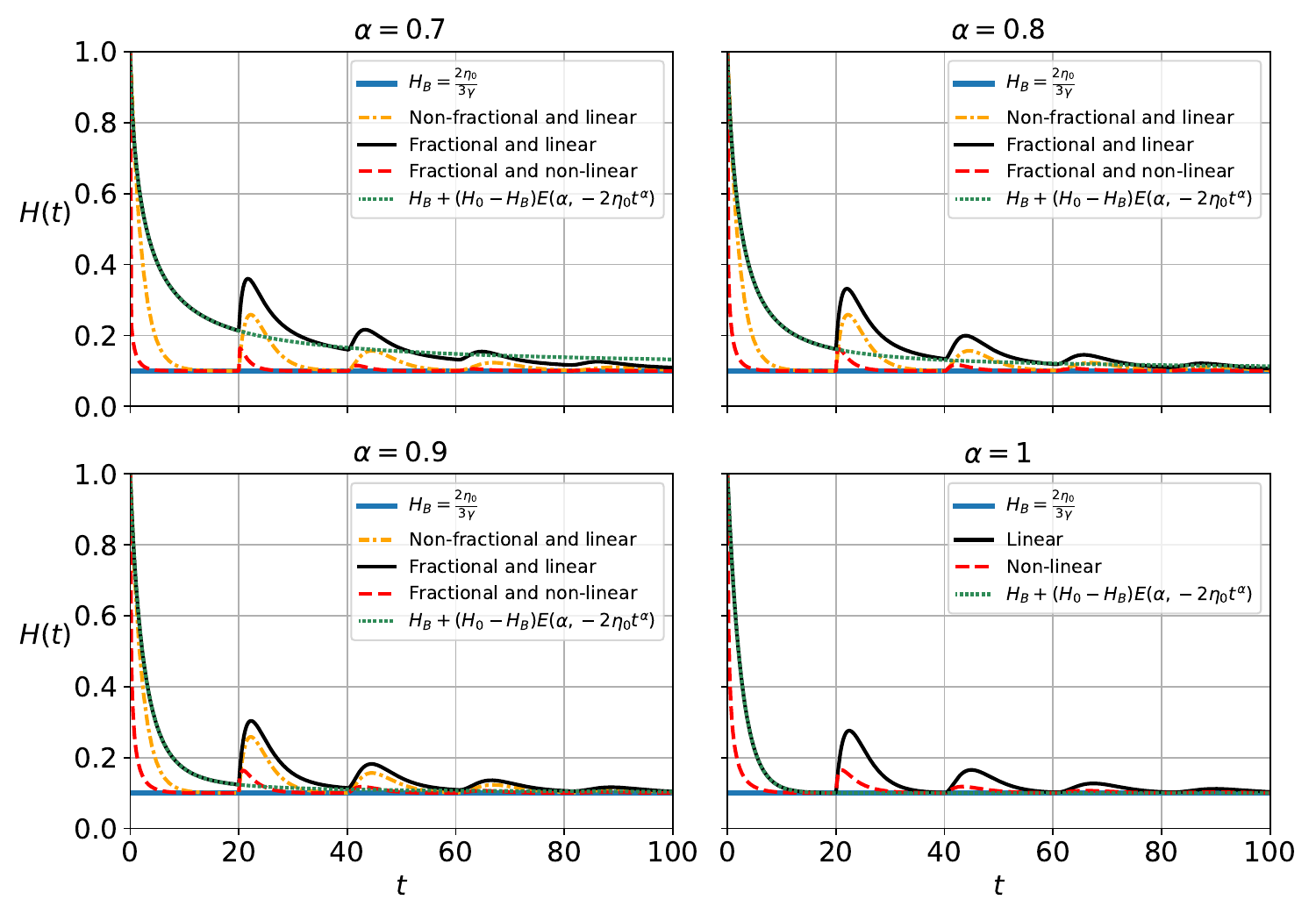}
\caption{Comparison between the non-fractional and linear equation \eqref{nonfractionalandretarded} (orange), fractional and linear equation \eqref{eqtosolve} (black) and fractional and non-linear equation \eqref{fractandnonlinear} (Red). The blue constant line is the de Sitter solution, and here $\eta_0=0.2$, $T=20$, $H_0=1$ and $\gamma=4/3$. The green and dashed lines represent the Mittag-Leffler function.}
\label{Fig2}
\end{figure}
\begin{figure}[h]
\centering
\includegraphics[width=1.0\textwidth]{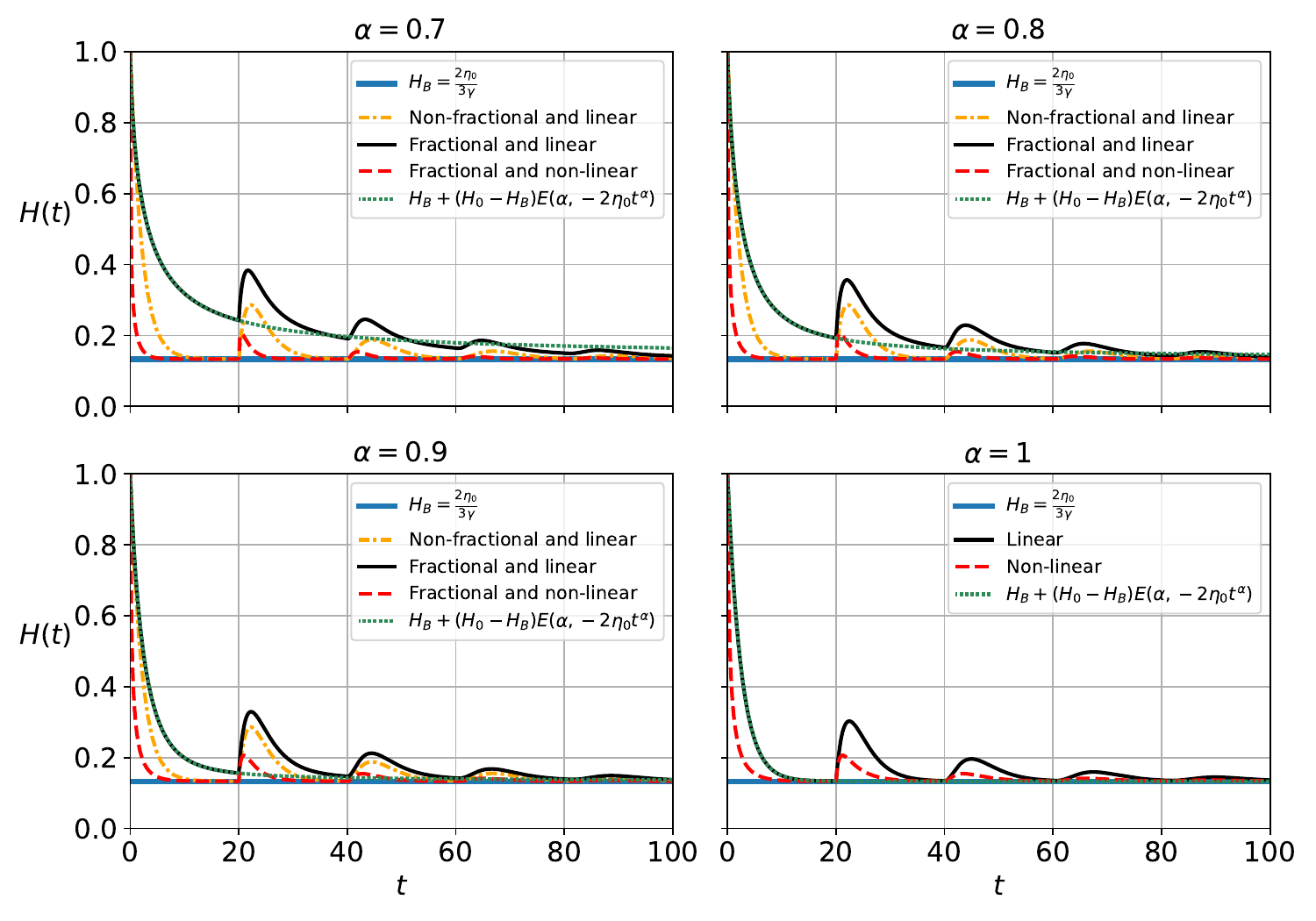}
\caption{Comparison between the non-fractional and linear equation \eqref{nonfractionalandretarded} (orange), fractional and linear equation \eqref{eqtosolve} (black) and fractional and non-linear equation \eqref{fractandnonlinear} (Red). The blue constant line is the de Sitter solution, and here $\eta_0=0.2$, $T=20$, $H_0=1$ and $\gamma=1$. The green and dashed lines represent the Mittag-Leffler function.}
\label{Fig2.1}
\end{figure}
\begin{figure}[h]
\centering
\includegraphics[width=0.5\textwidth]{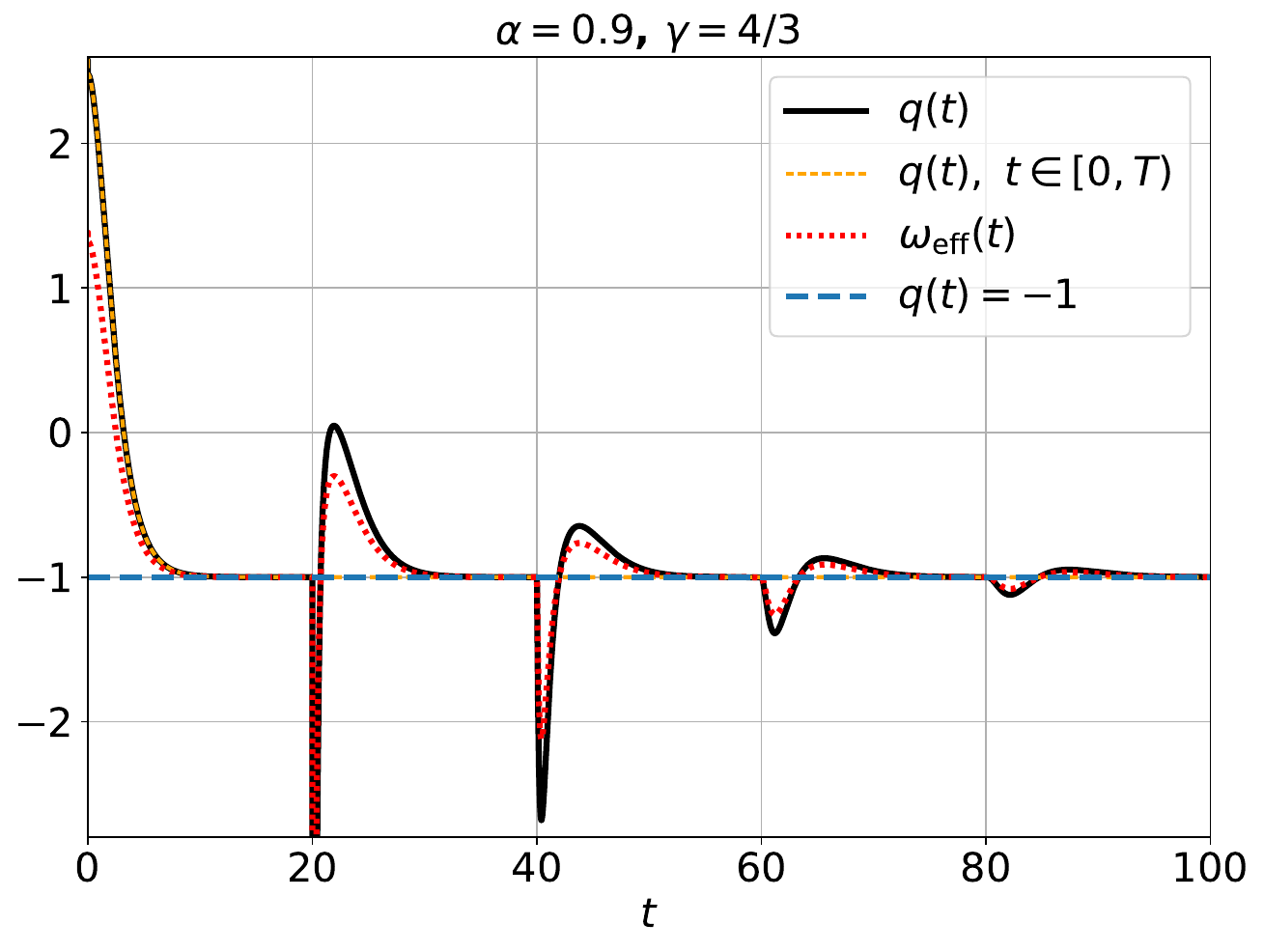}\includegraphics[width=0.5\textwidth]{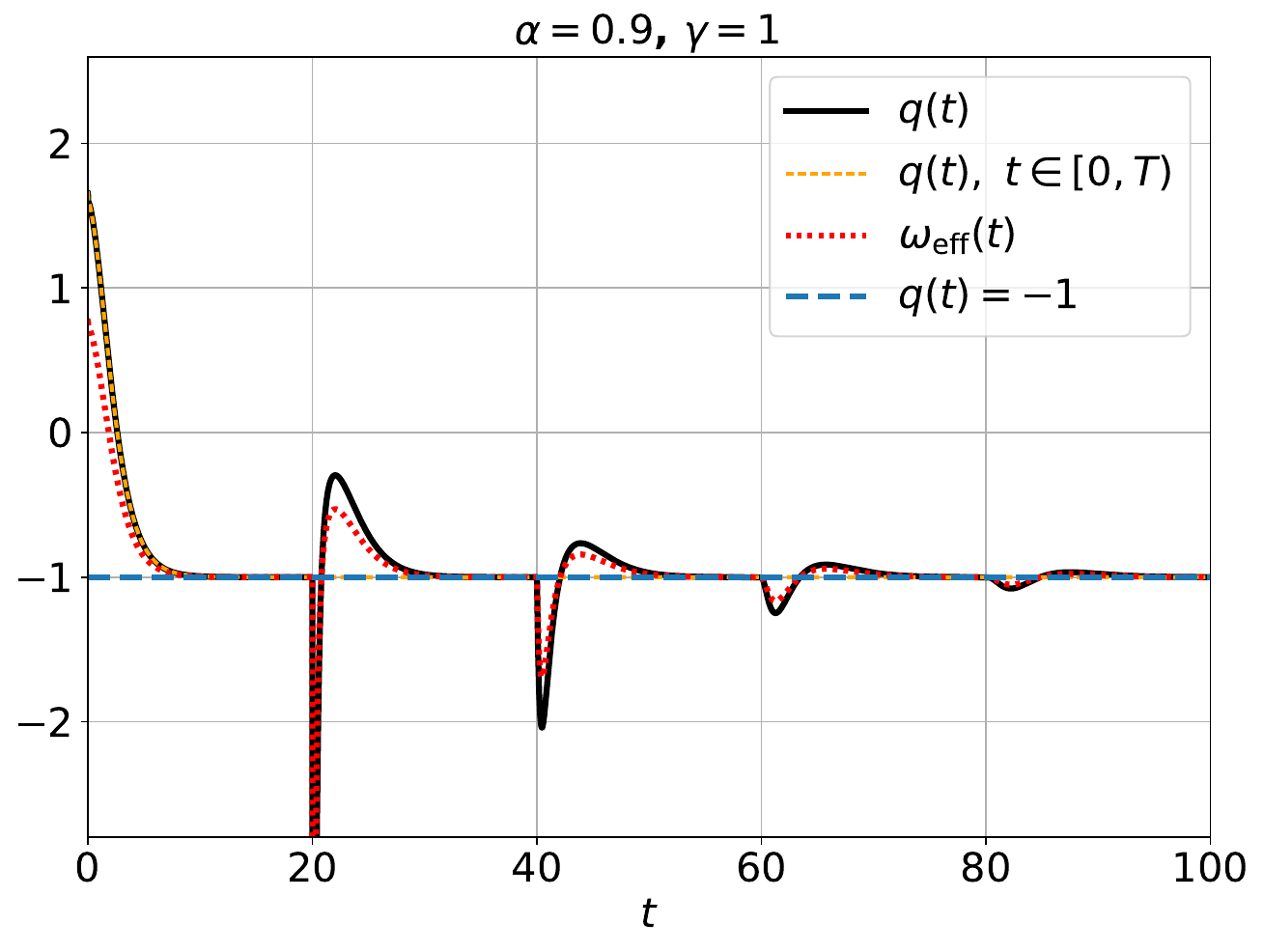}
\caption{\label{Fig0Obs0} Numerical solution of the functions $q(t)$ and  $\omega_\text{eff}(t)$ for non-linear case, for the cases $\gamma=4/3,1$, considering $\alpha=0.9$. The other parameters are $\eta_0=0.2$, $T=20$ and $H_0=1$. The minimum values of $q(t)$ and $\omega_\text{eff}(t)$ are $(q=-33.5,\omega_\text{eff}=-22.7)$ for the case $\gamma=4/3$ (radiation), and $(q=-18.6,\omega_\text{eff}=-12.7)$ for the case $\gamma=1$ (matter).}
\end{figure}

\subsection{Generalization}
For the sake of generality, we can consider the next equation: \begin{equation}\label{general-eq}
    \leftindex_{}^{\text{C}}D_t^\alpha y(t)=\sum_{r=0}^mc_ry(t-rT), \quad y(t)=0\quad \forall t<0,
\end{equation}
with the initial conditions:
\begin{equation}\label{ini_conds}
y(0) = y_0, \quad y'(0) = y_1, \quad \ldots, \quad y^{(n-1)}(0) = y_{n-1}, \quad n-1<\alpha < n.
\end{equation}
Using \eqref{Laplace-transform-of-Caputo} and \eqref{multipledelays},
and applying the Laplace Transform on equation \eqref{general-eq}, we obtain
\begin{equation}
    \mathcal{L}\left\{y(t)\right\}\left(s^\alpha-\sum_{r=0}^mc_re^{-rsT}\right)=\sum_{k=0}^{n-1}s^{\alpha-k-1}y^{(k)}(0).
\end{equation}
Then,
 
\begin{equation}
\begin{split}
    \mathcal{L}\left\{y(t)\right\}&=\frac{\sum_{k=0}^{n-1}s^{\alpha-k-1}y^{(k)}(0)}{s^\alpha-\sum_{r=0}^mc_re^{-rsT}}, \quad s^\alpha\neq \sum_{r=0}^mc_re^{-rsT}\\
    &=\frac{s^{\alpha-1}y^{(0)}(0)}{s^\alpha-\sum_{i=0}^mc_re^{-rsT}}+\frac{s^{\alpha-2}y^{(1)}(0)}{s^\alpha-\sum_{r=0}^mc_re^{-rsT}}+\cdots+\frac{s^{\alpha-n}y^{(n-1)}(0)}{s^\alpha-\sum_{r=0}^mc_re^{-rsT}}\\
    &=\frac{y^{(0)}(0)}{s^{1-\alpha}\left(s^\alpha-\sum_{r=0}^mc_ie^{-rsT}\right)}+\frac{y^{(1)}(0)}{s^{2-\alpha}\left(s^\alpha-\sum_{r=0}^mc_re^{-rsT}\right)}+\cdots+\frac{y^{(n-1)}(0)}{s^{n-\alpha}\left(s^\alpha-\sum_{r=0}^mc_re^{-rsT}\right)}\\
    &=\frac{y^{(0)}(0)}{s\left(1-s^{-\alpha}\sum_{r=0}^mc_re^{-rsT}\right)}+\frac{y^{(1)}(0)}{s^2\left(1-s^{-\alpha}\sum_{r=0}^mc_re^{-rsT}\right)}+\cdots+\frac{y^{(n-1)}(0)}{s^n\left(1-s^{-\alpha}\sum_{r=0}^mc_re^{-rsT}\right)}\\
    &=\frac{y^{(0)}(0)}{s}\sum_{j=0}^\infty \left(s^{-\alpha}\sum_{r=0}^mc_re^{-rsT}\right)^j+\frac{y^{(1)}(0)}{s^2}\sum_{j=0}^\infty \left(s^{-\alpha}\sum_{r=0}^mc_re^{-rsT}\right)^j \\ & +\cdots +\frac{y^{(n-1)}(0)}{s^n}\sum_{j=0}^\infty \left(s^{-\alpha}\sum_{r=0}^mc_re^{-rsT}\right)^j\\
    &=\sum_{k=0}^{n-1}\left[\frac{y^{(k)}(0)}{s^{k+1}}\sum_{j=0}^\infty \left(s^{-\alpha}\sum_{r=0}^mc_re^{-rsT}\right)^j\right]=\sum_{k=0}^{n-1}\sum_{j=0}^\infty \frac{y^{(k)}(0)}{s^{\alpha j+k+1}}\left(\sum_{r=0}^mc_re^{-rsT}\right)^j, \\ &\quad 0<\left|s^{-\alpha}\sum_{r=0}^mc_re^{-rsT}\right|<1.
\end{split}
\end{equation}

\begin{Remark}
    The convergence condition $0<\left|s^{-\alpha}\sum_{r=0}^mc_re^{-rsT}\right|<1$ is satisfied by all $s$ with $s^\alpha> m \max_r(c_r).$
\end{Remark}
But by the Multinomial Theorem, we have
\begin{equation}
    \left(x_0+x_1+\cdots+x_m\right)^n=\sum_{\substack{k_0 + k_1 + \cdots + k_m = n \\ k_0, k_1, \ldots, k_m \geq 0}}\frac{n!}{k_0!k_1!\cdots k_m!}x_0^{k_0}x_1^{k_1}\cdots x_m^{k_m}.
\end{equation}
Finally,
 
\begin{equation}
\begin{split}
    \mathcal{L}\left\{y(t)\right\}&=\sum_{k=0}^{n-1}\sum_{j=0}^\infty \frac{y^{(k)}(0)}{s^{\alpha j+k+1}}\left(c_0+c_1e^{-sT}+c_2e^{-2sT}+\cdots+c_me^{-msT}\right)^j\\
    &=\sum_{k=0}^{n-1}\sum_{j=0}^\infty \frac{y^{(k)}(0)}{s^{\alpha j+k+1}}\sum_{\substack{j_0 + j_1 + \cdots + j_m = j \\ j_0, j_1, \ldots, j_m \geq 0}}\frac{j!}{j_0!j_1!\cdots j_m!}c_0^{j_0}c_1^{j_1}e^{-j_1sT}c_2^{j_2}e^{-2j_2sT}\cdots c_m^{j_m}e^{-mj_msT}\\
    &=\sum_{k=0}^{n-1}\sum_{j=0}^\infty \sum_{\substack{j_0 + j_1 + \cdots + j_m = j \\ j_0, j_1, \ldots, j_m \geq 0}}\frac{j!}{j_0!j_1!\cdots j_m!}c_0^{j_0}c_1^{j_1}c_2^{j_2}\cdots c_m^{j_m}\frac{y^{(k)}(0)e^{-\left(j_1+2j_2+\cdots+mj_m\right)sT}}{s^{\alpha j+k+1}},
    \end{split}
    \end{equation}
    such that 
    \begin{equation}
\begin{split}
    y(t)&=\sum_{k=0}^{n-1}\sum_{j=0}^\infty \sum_{\substack{j_0 + j_1 + \cdots + j_m = j \\ j_0, j_1, \ldots, j_m \geq 0}}\frac{j!}{j_0!j_1!\cdots j_m!}c_0^{j_0}c_1^{j_1}c_2^{j_2}\cdots c_m^{j_m}y^{(k)}(0)\mathcal{L}^{-1}\left[\frac{e^{-\left(j_1+2j_2+\cdots+mj_m\right)sT}}{s^{\alpha j+k+1}}\right]\\
    &=\sum_{k=0}^{n-1}\sum_{j=0}^\infty \sum_{\substack{j_0 + j_1 + \cdots + j_m = j \\ j_0, j_1, \ldots, j_m \geq 0}}\frac{j!}{j_0!j_1!\cdots j_m!}c_0^{j_0}c_1^{j_1}c_2^{j_2}\cdots c_m^{j_m}y^{(k)}(0) \\
    & \times \frac{\left[t-(\sum_{r=1}^m r j_r)T\right]^{\alpha j+k}\theta(t-(\sum_{r=1}^m r j_r)T)}{\Gamma(\alpha j+k+1)}. \label{gen-y}
\end{split}
\end{equation}

\begin{Remark}\label{Rem_9}
    For each $t>0$, the inner series in equation \eqref{gen-y} is a finite sum. To see this, note that $\theta(t-(\sum_{r=1}^m r j_r)T)=0$ for all $(\sum_{r=1}^m r j_r)\geq t/T$. Then, 
     
    \begin{equation}
\begin{split}
    y(t)&=\sum_{k=0}^{n-1}\sum_{j=0}^\infty \sum_{\substack{j_0 + j_1 + \cdots + j_m = j \\ j_0, j_1, \ldots, j_m \geq 0}}^{\sum_{r=1}^m r j_r=\lfloor t/T\rfloor}\frac{j!}{j_0!j_1!\cdots j_m!}c_0^{j_0}c_1^{j_1}c_2^{j_2}\cdots c_m^{j_m}y^{(k)}(0) \\
    & \times \frac{\left[t-(\sum_{r=1}^m r j_r)T\right]^{\alpha j+k}\theta(t-(\sum_{r=1}^m r j_r)T)}{\Gamma(\alpha j+k+1)}. \label{Gen-y}
\end{split}
\end{equation}

    Furthermore, if we divide the time domain into intervals of length $T$, for each $t$, there exists an $n\in\mathbb{N}_0$, such that $t\in [nT,(n+1)T)$, and  $\lfloor t/T\rfloor=n$.
\end{Remark}
\begin{Proposition}\label{Prop_11}
 For $t>0$, the general solution of   \eqref{general-eq} with the initial conditions \eqref{ini_conds} is \eqref{Gen-y}. 
\end{Proposition}
\begin{proof}
    Proposition \ref{Prop_11} is proven by using the Remark \ref{Rem_9}.
\end{proof}
We can investigate the convergence of the partial sums

\begin{align}
    H_n(t)&= H_B+S_n(t), \label{Gen-H}
\\
    S_n(t)&=\sum_{k=0}^{n-1}\sum_{j=0}^{n} \sum_{\substack{j_0 + j_1 + \cdots + j_m = j \\ j_0, j_1, \ldots, j_m \geq 0}}^{\sum_{r=1}^m r j_r=\lfloor t/T\rfloor}\frac{j!}{j_0!j_1!\cdots j_m!}c_0^{j_0}c_1^{j_1}c_2^{j_2}\cdots c_m^{j_m}y^{(k)}(0)\\
    & \times \frac{\left[t-(\sum_{r=1}^m r j_r)T\right]^{\alpha j+k}\theta(t-(\sum_{r=1}^m r j_r)T)}{\Gamma(\alpha j+k+1)}. \label{Gen-Sn}
\end{align}
as $n \rightarrow \infty$. 
\begin{Remark}
\label{Prop_12}
For $t>0$, taking limit $n\rightarrow \infty$, 
\[H(t)=  \lim_{n\rightarrow \infty}  H_n(t)=H_B+ \lim_{n \rightarrow \infty} S_n(t),\] we obtain the solution 
\begin{equation}
     H(t)=H_B+ y(t), 
\end{equation} with $H_B=(2\eta_0)/(3\gamma)$ and  $y(t)$ defined by \eqref{Gen-y}
of the linearized equation
\begin{equation}
     \leftindex_{}^{\text{C}}D_t^\alpha  H(t)= \sum_{r=0}^mc_r \left(H (t-rT) - H_B\right).
\end{equation}
With the expression for $H(t)$ we can derive $a(t)$, $q(t)$ and $w_{\text{eff}}(t)$ by calculating 
\eqref{Togeta(t)}, \eqref{weffa} and \eqref{weff}.     
\end{Remark}
\section{Conclusions}\label{sect.5}

Fractional time-delayed differential equations (FTDDEs) bridged fractional calculus and time delays, providing advanced tools to model complex systems in fields such as cosmology. These systems included viscosity and fluid dynamics in space. Techniques like Laplace transforms and Mittag-Leffler functions proved essential in solving FTDDEs.
Viscous cosmology, emphasizing dissipative effects, offered novel insights into cosmic evolution while introducing practical modeling frameworks. 

Employing effective pressure terms for cosmic fluids unveiled mechanisms driving inflation and accelerated expansion without relying on scalar fields. Moreover, FTDDEs captured delayed responses in cosmic fluids, expanding their applications and thereby paving the way for innovative research.

We explored these concepts by studying simpler fractional time-delayed differential equations with linear responses. This foundation extended to first-order, fractional Caputo, and higher-order fractional differential equations with delays. Their characteristic equations were solved using Laplace transforms, providing valuable insights into real-world phenomena.

In the context of cosmology, we derived an equation from the Friedmann and continuity equations, incorporating a viscosity term linked to the delayed Hubble parameter. Fractional calculus advanced this into a fractional delayed differential equation, known as Caputo fractional derivatives, which were analytically solved for the Hubble parameter. Due to the complexity of the analytical solution, numerical representations were also developed. These solutions asymptotically converged to the de Sitter equilibrium point, representing a significant result in cosmological research. Additionally, solutions for FTDDEs with delays as multiples of a fundamental delay were analyzed, offering further extensions.
In summary, integrating fractional calculus, viscous cosmology, and time-delayed equations established a robust framework for addressing limitations in standard cosmological models. This interdisciplinary approach opened new research avenues and enriched our understanding of the Universe's fundamental properties.

%%%%%%%%%%%%%%%%%%%%%%%%%%%%%%%%%%%%%%%%%%
\vspace{6pt} 

%%%%%%%%%%%%%%%%%%%%%%%%%%%%%%%%%%%%%%%%%%

\section*{Author contributions}

{All the authors contributed to conceptualization; methodology; software; formal analysis; investigation;  writing---original draft preparation; writing---review and editing. 

The corresponding author also contributed to supervision, project administration, and funding acquisition. 
All authors have read and agreed to the published version of the manuscript.}

\section*{Funding}

{Bayron Micolta-Riascos, Genly Leon \& Andronikos Paliathanasis were funded by Agencia Nacional de Investigación y Desarrollo (AnID) through Proyecto Fondecyt Regular 2024,  Folio 1240514, Etapa 2025. They also thank Vicerrectoría de Investigación y Desarrollo Tecnológico (VRIDT) at Universidad Católica del Norte for support through núcleo de Investigación Geometría Diferencial y Aplicaciones (Resolución VRIDT n°096/2022 \& n°098/2022).}

\section*{Data availability}

{The data supporting this article can be found in Section~\ref{sec:numerics}.} 

\section*{Acknowledgments}

{Genly Leon would like to express his gratitude towards faculty member Alan Coley and staff members Anna Maria Davis, Nora Amaro, Jeanne Clyburne, and Mark Monk for their warm hospitality during the implementation of the final details of the research in the Department of Mathematics and Statistics at Dalhousie University. Genly Leon dedicates this work to his father, who sadly passed away.}

\section*{Conflicts of interest}

{We declare no conflict of interest. The funders had no role in the design of the study; in the collection, analyses, or interpretation of data; in the writing of the manuscript; or in the decision to publish the~results.}

\appendix

\section{Lambert $W$ Function}\label{app000}

The Lambert $W$ function, product logarithm, or $W$ function, is a set of functions denoted as \(W(x)\). The Lambert $W$ function satisfies the equation \( W(x) e^{W(x)} = x \) for any complex number \( x \). It has multiple branches, but the two most commonly used are the principal branch, \( W_0(x) \), which is real-valued for \( x \geq -1/e \), and the secondary branch, \( W_{-1}(x) \), which is real-valued for \( -1/e \leq x < 0 \). Moreover,  \( W(0) = 0 \) and \( W(-1/e) = -1 \).
The Lambert $W$ function is used in various fields, such as solving transcendental equations involving exponentials and logarithms, analyzing the behavior of specific dynamical systems, calculating the number of spanning trees in a complete graph, and modeling growth processes and delay differential equations.

\section{Mittag-Leffler functions}\label{app00}

From the Maclaurin series expansion of the exponential,
\begin{equation}
    e^z=\sum_{n=0}^{\infty}\frac{z^n}{n!},
    \label{1.1.2.1}
\end{equation}
We replace the factorial with the Gamma function:
\begin{equation}
    e^z=\sum _{n=0}^{\infty}\frac{z^n}{\Gamma (n+1)}.
    \label{1.1.2.2}
\end{equation}
That can then be extended as follows:
\begin{equation}
    E(\alpha,z)=\sum_{n=0}^{\infty}\frac{z^n}{\Gamma (\alpha n+1)},
    \label{1.1.2.3}
\end{equation}
where $\alpha\in \mathbb{C}$ is arbitrary real with $\Re(\alpha)>0$.

In fractional calculus, this function is of similar importance to the exponential function in standard calculus. For some values of $\alpha$ and functions of $z$, already known functions can be obtained:
\begin{equation}
    E(2,-z^2)=\cos{z}, \hspace{20px} E(1/2,z^{1/2})=e^z\left[1+\text{erf}( z^{1/2})\right],
    \label{1.1.2.4}
\end{equation}
where the Error function, $\text{erf}(z)$, is given by
\begin{equation}
    \text{erf}(z)=\frac{2}{\sqrt{\pi}}\int_{0}^{z}e^{-t^2}dt.
    \label{1.1.2.5}
\end{equation}
The Mittag-Leffler function can also be extended as follows:
\begin{equation}
    E(\alpha,\beta,z)=\sum _{n=0}^{\infty}\frac{z^n}{\Gamma (\alpha n+\beta)},
    \label{1.1.2.6}
\end{equation}
which is known as the generalized Mittag-Leffler function and has several special cases, e.g.
\begin{equation}
    E(1,2,z)=(e^z-1)/z, \hspace{20px} E(2,2,z^2)=\sinh{(z)}/z.
    \label{1.1.2.7}
\end{equation}

\section{Laplace transform of the time-delayed function}
\label{app4}

We can find an analytical solution using the Laplace transform of a function $y(t)$ as given by~\cite{Schiff:2013}
\begin{equation}
    \mathcal{L}\left\{y(t)\right\}=\int_0^\infty y(t)e^{-st}\,dt,\quad s>0.
\end{equation}
We must also know the Laplace transform of a delayed function:
\begin{equation}
    \mathcal{L}\left\{y(t-T)\right\} =\int_0^\infty y(t-T)e^{-st}\,dt=\int_{-T}^\infty y(u)e^{-s(u+T)}\,du =e^{-sT}\int_{-T}^\infty y(t)e^{-st}\,dt.
\end{equation}
We assume that $y(t)=0 \quad\forall t<0$, which are saying us that $y(t)$ does not have an history for $t<0$, thus,
\begin{equation}\nonumber
\begin{split}
    \mathcal{L}\left\{y(t-T)\right\}&=e^{-sT}\int_0^\infty y(t)e^{-st}\,dt =e^{-sT}\mathcal{L}\left\{y(t)\right\}.
\end{split}
\end{equation}
Therefore, we have
\begin{equation}\label{LaplaceTdelayedf}
    \mathcal{L}\left\{y(t-T)\right\}=e^{-sT}\mathcal{L}\left\{y(t)\right\}.
\end{equation}
Applying recursively, 
\begin{equation}
    \mathcal{L}\left\{y(t-rT)\right\}=e^{-rsT}\mathcal{L}\left\{y(t)\right\}, r=1,2,\ldots m. \label{multipledelays}
\end{equation}

\section{Laplace transform of the Caputo derivative}
\label{app2}

The Caputo derivative is defined as  
\begin{equation}\label{CaputoD}
    \leftindex_{}^{\text{C}}D_t^\alpha y(t)=\frac{1}{\Gamma(n-\alpha)}\int_0^t\frac{d^n y(\tau)}{d\tau^n}\cdot(t-\tau)^{n-1-\alpha}\,d\tau,\quad n-1<\alpha<n.
\end{equation}  
Then, we can calculate the Laplace transform of the Caputo derivative:  
\begin{equation}
\begin{split}
    \mathcal{L}\left\{\leftindex_{}^{\text{C}}D_t^\alpha y(t)\right\}&=\frac{1}{\Gamma(n-\alpha)}\int_0^\infty \left\{\int_0^t\frac{d^n y(\tau)}{d\tau^n}\cdot(t-\tau)^{n-1-\alpha}\,d\tau\right\}e^{-st}\,dt\\
    &=\frac{1}{\Gamma(n-\alpha)}\int_0^\infty \left(\int_\tau^\infty\frac{d^n y(\tau)}{d\tau^n}\cdot(t-\tau)^{n-1-\alpha}\cdot e^{-st}\,dt\right)\,d\tau\\
    &=\frac{1}{\Gamma(n-\alpha)}\int_0^\infty \frac{d^n y(\tau)}{d\tau^n} \left(\int_\tau^\infty(t-\tau)^{n-1-\alpha}e^{-st}\,dt\right)\,d\tau\\
    &=\frac{1}{\Gamma(n-\alpha)}\int_0^\infty \frac{d^n y(\tau)}{d\tau^n} \left(\int_0^\infty u^{n-1-\alpha}e^{-s(u+\tau)}\,du\right)\,d\tau\\
    &=\frac{1}{\Gamma(n-\alpha)}\int_0^\infty \frac{d^n y(\tau)}{d\tau^n}e^{-s\tau} \left(\int_0^\infty u^{n-1-\alpha}e^{-su}\,du\right)\,d\tau\\
    &=\frac{1}{\Gamma(n-\alpha)}\left(\int_0^\infty\frac{d^n y(\tau)}{d\tau^n}e^{-s\tau}\,d\tau\right)\left(\int_0^\infty u^{n-1-\alpha}e^{-su}\,du\right).
\end{split}
\end{equation}  
Recalling the definition of the Gamma function \eqref{1.1.1.2}, 
we find that  
\begin{equation}
    \mathcal{L}\left[\leftindex_{}^{\text{C}}D_t^\alpha y(t)\right]=\frac{\displaystyle\left(\int_0^\infty\frac{d^n y(\tau)}{d\tau^n}e^{-s\tau}\,d\tau\right)\left(\int_0^\infty u^{n-1-\alpha}e^{-su}\,du\right)}{\displaystyle\int_0^\infty t^{n-1-\alpha}e^{-t}\,dt}.
\end{equation}   
Then,  
\begin{equation}
\begin{split}
    \mathcal{L}\left[\leftindex_{}^{\text{C}}D_t^\alpha y(t)\right]
    &=s^{\alpha-n}\int_0^\infty \frac{d^ny(\tau)}{d\tau^n}e^{-s\tau}\,d\tau =s^{\alpha-n}\mathcal{L}\left\{\frac{d^ny(t)}{dt^n}\right\}.
\end{split}
\end{equation}  
Therefore,  
\begin{equation}\label{LaplaceRelation 1}
    \mathcal{L}\left\{\leftindex_{}^{\text{C}}D_t^\alpha y(t)\right\}=s^{\alpha-n}\mathcal{L}\left\{\frac{d^ny(t)}{dt^n}\right\}.
\end{equation}  
Integrating $n$ times by parts, the Laplace transform of the $n$-th derivative is:
\begin{equation}
    \mathcal{L}\left\{\frac{d^ny(t)}{dt^n}\right\}=s^n\mathcal{L}\left\{y(t)\right\}-\sum_{k=0}^{n-1}s^{n-k-1}y^{(k)}(0).
\end{equation}
Finally, the Laplace transform of Caputo's derivative is given by \cite{mathai2008special}
\begin{equation}\label{Laplace-transform-of-Caputo}
    \mathcal{L}\left\{\leftindex_{}^{\text{C}}D_t^\alpha y(t)\right\}:=s^{\alpha-n}\mathcal{L}\left\{\frac{d^ny(t)}{dt^n}\right\}=s^\alpha \mathcal{L}\left\{y(t)\right\}-\sum_{k=0}^{n-1}s^{\alpha-k-1}y^{(k)}(0),\quad n-1<\alpha<n.
\end{equation}

\section{Numerical Considerations and Forward Difference Formulation}
\label{appFDF}
As explained in \cite{herrmann2014fractional}, following Weierstraß, the first derivative is defined as:
\begin{equation}
v(t) = \frac{d}{dt} x(t) =
\begin{cases}
\lim\limits_{h \to 0} \frac{x(t) - x(t-h)}{h}, & \text{backward difference} \\
\lim\limits_{h \to 0} \frac{x(t+h) - x(t-h)}{2h}, & \text{central difference} \\
\lim\limits_{h \to 0} \frac{x(t+h) - x(t)}{h}, & \text{forward difference, } h < \epsilon\ll 1, h \geq 0
\end{cases}
\end{equation}
If \( x(t) \) is smooth and analytic, the different derivative formulations are equivalent when \( t \) represents a time-like coordinate. The forward difference method is advantageous in time-dependent simulations where only past values are available for computing derivatives. Unlike central differences, which require values from both past and future points, forward differences rely solely on prior data, making them well-suited for real-time calculations and delay differential equations. Additionally, they tend to be numerically stable in specific discretized systems by avoiding dependence on future values, reducing complexity. While forward differences provide a reasonable first-order approximation for smooth functions, they may introduce truncation errors. In contrast, central differences use both \( x(t+h) \) and \( x(t-h) \), while backward differences rely exclusively on past values.

Various physical quantities, such as velocity, current, and flux density, follow local principles and obey differential equations. The derivative can also be expressed in its integral form:
\begin{equation}
v(t) = \frac{d}{dt} x(t) =
\begin{cases}
2 \int_{0}^{\infty} dh \, \delta(h) v(t - h), & \text{backward} \\
2 \int_{0}^{\infty} dh \, \delta(h) \frac{v(t+h) + v(t-h)}{2}, & \text{central} \\
2 \int_{0}^{\infty} dh \, \delta(h) v(t + h), & \text{forward}
\end{cases},
\end{equation}
where \( \delta(t) \) is the Dirac delta function, satisfying \cite{Lighthill_1958}:
\begin{equation}
\int_{-\infty}^{\infty} dh \, \delta(h) f(h) = f(0).
\label{dirac_property}
\end{equation}

The Dirac delta function is defined as a limit of equivalent regular sequences of good functions:
\begin{equation}
\delta(t) = \lim_{a \to 0} w(a, t),
\label{dirac_definition}
\end{equation}
\noindent
where \( a \) is a smooth parameter satisfying \( a \geq 0 \).

Following \cite{Lighthill_1958}, as cited in \cite{herrmann2014fractional}, a good function is infinitely differentiable and satisfies:
\[
O(|t|^{-n}) \quad \text{as } |t| \to \infty, \quad \forall n.
\]

Additionally, a sequence \( w(a, t) \) of good functions is called regular if, for any good function \( f(t) \), the following limit exists:
\begin{equation}
\lim\limits_{a \to 0} \int_{-\infty}^{+\infty} w(a, t) f(t) \, dt.
\label{regular_limit}
\end{equation}

Finally, two regular sequences of good functions are equivalent if the limit is the same for both sequences.
Typical examples of such sequences include:
\begin{equation}
w(a, t) =
\begin{cases}
\exp\left(-\frac{|t|}{a}\right)^p, \quad p \in \mathbb{R}^{+}, & \text{exponential} \\
Ai\left(\frac{|t|}{a}\right), & \text{Airy function} \\
\frac{1}{t} \sin\left(\frac{|t|}{a}\right), & \text{sine function}
\end{cases}.
\end{equation}

\section{Optimized algorithms}\label{appF}

\subsection{Optimized algorithm to implement exact solutions \eqref{AnalyticalHP1},  \eqref{q-ini}-together with \eqref{q_2}
and \eqref{w-ini}-together with \eqref{weff_2}, given in Section \ref{SECT:3.1}}
\label{appF1}

In this appendix, we present an optimized algorithm that implements the exact solutions \eqref{AnalyticalHP1}, \eqref{q-ini} combined with \eqref{q_2}, and \eqref{w-ini} combined with \eqref{weff_2}, as described in Section \ref{SECT:3.1}. This analytical solution is constructed using a finite series that is summed computationally. For the numerical representations shown in Figures \ref{Fig00} and \ref{Fig0ObsB}, we used $5000$ terms from the partial sums to compute integrations over the interval $t \in [0, 100]$. To integrate over a generic interval $t \in [0, \text{IntervalCount} \cdot h]$, where $h$ represents the time step size, we require $\text{IntervalCount}$ terms.

\subsubsection{Algorithm}

\textbf{Inputs}
\begin{itemize}
    \item $T$: Length of delay
    \item $m$: number of sub-intervals per delay interval
    \item $\eta_0$: Parameter $\eta_0$
    \item $\gamma$: Parameter $\gamma$
    \item $H_0$: Initial condition for $H$
    \item $\text{IntervalCount}$: number of intervals
\end{itemize}

\textbf{Derived Inputs}
\begin{itemize}
    \item Time step size:
    \[
    h = \frac{T}{m}
    \]
    \item Total number of steps:
    \[
    \text{TotalSteps} = m \times \text{IntervalCount}
    \]
    \item Constant $H_B$:
    \[
    H_B = \frac{2 \cdot \eta_0}{3 \cdot \gamma}
    \]
    \item Set of points:
    \[
    t_k = \{k \cdot h \, | \, k = 0, 1, \dots, \text{TotalSteps}\}
    \]
\end{itemize}

\textbf{Initialization}
\begin{itemize}
    \item Define scalar function $H(t)$:
    \[
    H(t) = 
    \begin{cases} 
        H_B + (H_0 - H_B) \cdot \exp(-2 \cdot \eta_0 \cdot t), & t < T \\\\
        H_B + (H_0 - H_B) \cdot \text{accumulator}, & t \geq T 
    \end{cases}
    \]
    where:
    \[
    \text{accumulator} = \sum_{k=0}^{\lfloor t / T \rfloor} \frac{\exp(-2 \cdot \eta_0 \cdot (t - k \cdot T)) \cdot (\eta_0 \cdot (t - k \cdot T))^k}{\Gamma(k + 1)}
    \]

    \item Define function $q(t)$:
    \[
    q(t) = 
    \begin{cases} 
        -1 - \frac{6 \cdot \gamma \cdot \eta_0 \cdot \exp(2 \cdot \eta_0 \cdot t) \cdot (2 \cdot \eta_0 - 3 \cdot \gamma \cdot H_0)}{\left(3 \cdot \gamma \cdot H_0 + 2 \cdot \eta_0 \cdot (\exp(2 \cdot \eta_0 \cdot t) - 1)\right)^2}, & t < T \\\\
        -1 + \frac{3 \cdot \gamma}{2} - \frac{\eta_0 \cdot H(t - T)}{H(t)^2}, & t \geq T 
    \end{cases}
    \]
\end{itemize}

\textbf{Calculate $w_{\text{eff}}$}
\begin{itemize}
    \item Compute $w_{\text{eff}}(t)$ using:
    \[
    w_{\text{eff}}(t) = \frac{2 \cdot q(t) - 1}{3}
    \]
\end{itemize}

\textbf{Process Results}
\begin{itemize}
    \item Iterate over $t_k$ and evaluate:
    \begin{itemize}
        \item $H(t)$
        \item $q(t)$
        \item $w_{\text{eff}}(t)$
    \end{itemize}
    \item Store results in a structured array:
    \[
    \text{results} = \{(t, H(t), q(t), w_{\text{eff}}(t)) \, | \, t \in t_k\}
    \]
\end{itemize}

\textbf{Output Results}
\begin{itemize}
    \item Plot $H(t)$ vs $t$:
    \[
    \text{ListLinePlot}[H(t)]
    \]
    \item Plot $q(t)$ vs $t$:
    \[
    \text{ListLinePlot}[q(t)]
    \]
    \item Plot $w_{\text{eff}}(t)$ vs $t$:
    \[
    \text{ListLinePlot}[w_{\text{eff}}(t)]
    \]
    \item Generate combined plot:
    \[
    \text{Show}[H(t), q(t), w_{\text{eff}}(t)]
    \]
\end{itemize}

\subsection{Algorithm for computing \( H(t) \) using equation \eqref{H_summation} with the mollifier \eqref{mollifier}.}

\label{sect_mollifier}

\textbf{Inputs}
\begin{itemize}
    \item $T$: Length of delay
    \item $\eta_0$: Parameter $\eta_0$
    \item $\gamma$: Parameter $\gamma$
    \item $H_0$: Initial condition for $H$
    \item $\text{Tolerance}$: Stopping condition threshold
    \item $t_{\text{Limit}}$: Time limit $t$
    \item $m$: Number of terms in partial summation $H_m(t)$
\end{itemize}
\newpage
\textbf{Derived Inputs}
\begin{itemize}
    \item Upper bound function:
    \[
    t(n) = \frac{\alpha (n+1)}{e (2\eta_0)^{\frac{1}{\alpha}}}.
    \]
 $\forall t<  t(n) $, $|H_0 - H_B| \cdot \frac{1}{\sqrt{2 \pi \alpha (n+1)}} \cdot \left(\frac{t}{t(n)}\right)^{\alpha (n+1)} \rightarrow 0$ and $t(n)\rightarrow \infty$ as $n\rightarrow\infty$.

 For a given $n$, the time to achieve the tolerance is 
\[t_{\text{Tol}}(n)=  t(n) \left(\frac{\sqrt{2 \pi \alpha (n+1)}\text{Tolerance}}{ |H_0 - H_B|}\right)^{\frac{1}{\alpha (n+1)}}\]

    \item Constant $H_B$:
    \[
    H_B = \frac{2 \cdot \eta_0}{3 \cdot \gamma}
    \]
\end{itemize}

\textbf{Initialization}
\begin{itemize}
    \item Set initial value: $n = 0$
    \item Compute $t_{\text{Final}} = \max\{t(n), t_{\text{Tol}}(n)\}$
    \item While the stopping condition is not met:
    \[t_{\text{Final}}< t_{\text{Limit}}\quad \text{and} \quad |H_0 - H_B| \cdot \frac{1}{\sqrt{2 \pi \alpha (n+1)}} \cdot \left(\frac{t}{t(n)}\right)^{\alpha (n+1)} > \text{Tolerance}\]
    \item Increment $n$ and update $t_{\text{Final}}$
    \item Set $n = m$ for the truncated sum $H_m(t)$
\end{itemize}
\textbf{Compute $H_n(t)$ and $H_m(t)$}
\begin{itemize}
    \item Compute recursive summation:
    \[
    H_n(t) = H_B + (H_0 - H_B) \cdot \sum_{k=0}^{\lfloor t / T \rfloor} \sum_{j=k}^{n} \frac{j!}{k!(j-k)!} (-2)^{j-k} \eta_0^j \frac{(t-kT)^{\alpha j} \theta(t-kT)}{\Gamma(\alpha j+1)}
    \]
    \item Evaluate $H_m(t)$
    \[
    H_m(t) = H_B + (H_0 - H_B) \cdot \sum_{k=0}^{\lfloor t / T \rfloor} \sum_{j=k}^{m} \frac{j!}{k!(j-k)!} (-2)^{j-k} \eta_0^j \frac{(t-kT)^{\alpha j} \theta(t-kT)}{\Gamma(\alpha j+1)}
    \]
\end{itemize}

\textbf{Apply Smooth Correction}
\begin{itemize}
    \item Compute smoothing factor:
    \[
    S(t, t_{\text{Final}}) = 1 - \exp\left(-\left(\frac{t - t_{\text{Final}}}{T}\right)^2\right)
    \]
    \item Compute correction term:
    \[
    C(t, n) = S(t, t_{\text{Final}}) \cdot (H_0 - H_B) (-1)^{n+1} \frac{1}{\sqrt{2\pi \alpha (n+1)}} \left[\frac{e (2\eta_0)^{1/\alpha} \cdot t}{\alpha(n+1)} \right]^{\alpha (n+1)}
    \]
    \item Apply correction for $t > t_{\text{Final}}$:
    \[
    H_n^{\text{corrected}}(t) = H_n(t) + C(t, n)
    \]
\end{itemize}

\textbf{Process Results}
\begin{itemize}
    \item Iterate over a range of $t_k$:
    \[
    t_k = \{k \cdot h \mid k = 0, 1, \dots, \lfloor t_{\text{Final}}/h \rfloor \}
    \]
    \item Compute values:
    \begin{itemize}
        \item $H_n^{\text{corrected}}(t_k)$
        \item $H_m(t_k)$
        \item Error:
        \[
        E(t_k) = |H_m(t_k) - H_n^{\text{corrected}}(t_k)|
        \]
    \end{itemize}
    \item Store results:
    \[
    \text{Results} = \{(t, H_n^{\text{corrected}}(t), H_m(t), E(t)) \mid t \in t_k\}
    \]
\end{itemize}

\textbf{Output Results}
\begin{itemize}
    \item Plot $H_n^{\text{corrected}}(t)$ vs. $t$:
    \[
    \text{ListLinePlot}[H_n^{\text{corrected}}(t)]
    \]
    \item Plot truncated sum $H_m(t)$ vs. $t$:
    \[
    \text{ListLinePlot}[H_m(t)]
    \]
    \item Plot error function $E(t)$ vs. $t$:
    \[
    \text{ListLinePlot}[E(t)]
    \]
    \item Generate stacked visualization:
    \[
    \text{Column} \left[ H_n^{\text{corrected}}(t), H_m(t), E(t) \right]
    \]
\end{itemize}

\subsection{Algorithm implementing the numerical procedure \eqref{algorithm-1}}
\label{AlgorithmB}
\textbf{Inputs:}
\begin{itemize}
    \item $\alpha$: Fractional order
    \item $T$: Length of delay
    \item $m$: number of sub-intervals per delay interval
    \item $\eta_0$: Parameter $\eta_0$
    \item $\gamma$: Parameter $\gamma$
    \item $H_0$: Initial condition for $H$
    \item $c_1 = -2 \eta_0$: Coefficient $c_1$
    \item $c_2 = \eta_0$: Coefficient $c_2$
    \item $\text{IntervalCount}$: number of intervals
\end{itemize}
\textbf{Derived Inputs:}
\begin{itemize}
    \item $h = \frac{T}{m}$: Time step size
    \item $\text{TotalSteps} = m \times \text{IntervalCount}$: Total number of time steps
\end{itemize}
\textbf{Initialization:}
\begin{itemize}
    \item Initialize empty lists: $y = \{\}$, $H = \{\}$, $q = \{\}$, $w_{\text{eff}} = \{\}$
    \item Compute repeated constants: 
    \[
    \text{constTerm} = \frac{2 \eta_0}{3 \gamma}, \quad \Gamma_\alpha = \Gamma(\alpha + 1)
    \]
    \item Set initial condition for $y$:
    \[
    y[1] = H_0 - \frac{2 \eta_0}{3 \gamma}
    \]
    Append $y[1]$ to list $y$
    \item Set initial condition for $H$:
    \[
    H[1] = H_0.
    \]
    Append $H[1]$ to list $H$
    \item Set initial condition for $q$:
    \[
    q[1] = -1 + \frac{(2 \eta_0 - 3 \gamma H[1]) \left(-1 + E(\alpha, -2 \eta_0 h^\alpha)\right)}{3 \gamma h H[1]^2}
    \]
    Append $q[1]$ to list $q$
      \item Set initial condition for $q$:
    \[
    w_{\text{eff}}[1] =  \frac{2 q[1] - 1}{3}
    \]
    Append $w_{\text{eff}}[1]$ to list $w_{\text{eff}}$
\end{itemize}
\textbf{First Interval Computation:}
For $k = 1$ to $m$:
\[
y[k] = y[1] \cdot E(\alpha, -2 \eta_0 (kh)^\alpha)
\]
\[
H[k] = \text{constTerm} + y[k]
\]
\[
q[k] = -1 + \frac{6 \gamma \eta_0 (kh)^{\alpha - 1} (3 \gamma H_0 - 2 \eta_0) E(\alpha, \alpha, -2 (kh)^\alpha \eta_0)}{\left[2 \eta_0 + (3 \gamma H_0 - 2 \eta_0) E(\alpha, -2 (kh)^\alpha \eta_0)\right]^2}
\]
\[
w_{\text{eff}}[k] = \frac{2 q[k] - 1}{3}
\]
Append values of $y[k]$, $H[k]$, $q[k]$, and $w_{\text{eff}}[k]$ to their respective lists.

\textbf{Subsequent Intervals Computation:}
For $n = m+1$ to $\text{TotalSteps}$:
\[
y[n+1] = y[n] + \frac{h^\alpha}{\Gamma_\alpha} \left(c_1 y[n] + c_2 
y[n-m]\right)
\]
\[
H[n] = \text{constTerm} + y[n]
\]
\[
q[n] = -1 - \frac{h^{\alpha - 1} \left(c_1 y[n] + c_2 y[n-m]\right)}{\Gamma_\alpha H[n]^2}
\]
\[
w_{\text{eff}}[n] = \frac{2 q[n] - 1}{3}
\]
Append values of $y[n+1]$, $H[n]$, $q[n]$, and $w_{\text{eff}}[n]$ to their respective lists.

\textbf{Output Results:}
Return the lists $y$, $H$, $q$, and $w_{\text{eff}}$ as the solution.

\subsection{Algorithm implementing the fractional nonlinear scheme \eqref{algorithm-2}}
\label{AlgorithmC}

\textbf{Inputs:}
\begin{itemize}
    \item $\alpha$: Fractional order
    \item $T$: Length of delay
    \item $m$: number of sub-intervals per delay interval
    \item $\eta_0$: Parameter $\eta_0$
    \item $\gamma$: Parameter $\gamma$
    \item $H_0$: Initial condition for $H$
    \item $c_1 = -2\eta_0$: Coefficient $c_1$
    \item $c_2 = \eta_0$: Coefficient $c_2$
    \item $\text{IntervalCount}$: number of intervals
\end{itemize}

\textbf{Derived Inputs:}
\begin{itemize}
    \item $h = \frac{T}{m}$: Time step size
    \item $\text{TotalSteps} = m \times \text{IntervalCount}$: Total number of time steps
\end{itemize}
\textbf{Initialization:}
\begin{itemize}
    \item Initialize empty lists: $y = \{\}$, $H = \{\}$, $q = \{\}$, $w_{\text{eff}} = \{\}$
    \item Compute repeated constants: 
    \[
    \text{constTerm} = \frac{2 \eta_0}{3 \gamma}, \quad \Gamma_\alpha = \Gamma(\alpha + 1)
    \]
    \item Set initial condition for $y$:
    \[
    y[1] = H_0 - \frac{2 \eta_0}{3 \gamma}
    \]
    Append $y[1]$ to list $y$
    \item Set initial condition for $H$:
    \[
    H[1] = H_0
    \]
    Append $H[1]$ to list $H$
     \item Set initial condition for $q$:
    \[
    q[1] =  -1 - \frac{h^{\alpha-1} \left( -\frac{3\gamma}{2} y[1]^2 - 2\eta_0 y[1] \right)}{\Gamma_\alpha H[1]^2}
    \]
    Append $q[1]$ to list $q$.
      \item Set initial condition for $q$:
    \[
    w_{\text{eff}}[1] =  \frac{2 q[1] - 1}{3}
    \]
    Append $w_{\text{eff}}[1]$ to list $w_{\text{eff}}$.
\end{itemize}
\textbf{First Interval Computation:}
For $k = 1$ to $m$:
\begin{align*}
    y[k+1] &= y[k] + \frac{h^\alpha}{\Gamma_\alpha} \left( -\frac{3\gamma}{2} y[k]^2 - 2\eta_0 y[k] \right) \\
    H[k] &= \text{constTerm} + y[k] \\
    q[k] &= -1 - \frac{y[k+1] - y[k]}{h H[k]^2} \\
    w_{\text{eff}}[k] &= \frac{2q[k] - 1}{3}
\end{align*}
Append values of $y[k+1]$, $H[k]$, $q[k]$, and $w_{\text{eff}}[k]$ to their respective lists.

\textbf{Delayed Recurrence Procedure for Subsequent Intervals:}
For $n = m+1$ to $\text{TotalSteps}$:
\begin{align*}
    y[n+1] &= y[n] + \frac{h^\alpha}{\Gamma_\alpha} \left( -\frac{3\gamma}{2} y[n]^2 - 2\eta_0 y[n] + \eta_0 y[n-m] \right) \\
    H[n] &= \text{constTerm} + y[n], \\
    q[n] &= -1 - \frac{h^{\alpha-1} \left( -\frac{3\gamma}{2} y[n]^2 - 2\eta_0 y[n] + \eta_0 y[n-m] \right)}{\Gamma_\alpha H[n]^2} \\
    w_{\text{eff}}[n] &= \frac{2q[n] - 1}{3}
\end{align*}
Append values of $y[n+1]$, $H[n]$, $q[n]$, and $w_{\text{eff}}[n]$ to their respective lists.

\textbf{Output Results:}
Return the lists $y$, $H$, $q$, and $w_{\text{eff}}$ as the solution.

\section*{References}

%\begin{thebibliography}{999}
\bibliography{main}
%\end{thebibliography}

\end{document}